\documentclass[twocolumn]{autart}                        
\usepackage{etoolbox}
\usepackage{cite}
\usepackage{graphicx}
\usepackage{algorithm}
\let\classAND\AND
\let\AND\relax
\usepackage{algorithmic}

\let\AND\classAND
\AtBeginEnvironment{algorithmic}{\let\AND\algoAND}
\usepackage{caption}
\usepackage{subcaption}
\usepackage{tikz}
\usetikzlibrary{shapes.geometric, arrows}
\usepackage{xcolor}
\usepackage[english]{babel}
\usepackage{amsmath,amssymb,amsfonts}
\newtheorem{theorem}{Theorem}
\newtheorem{lemma}[theorem]{Lemma}

\tikzstyle{process} =[rectangle, minimum width=\linewidth, minimum height=1cm, text centered, text width=8cm, draw=black, fill=grey!5]
\tikzstyle{process2} =[rectangle, minimum width=\linewidth, minimum height=1cm, text centered, text width=8cm, draw=black, fill=grey!5]
\tikzstyle{process3} =[rectangle, minimum width=\linewidth, minimum height=1cm, text centered, text width=8.4cm, draw=black, fill=grey!5]
\begin{document}

\begin{frontmatter}

\title{Reinforcement Learning-based Control of Nonlinear Systems using Carleman Approximation: Structured and Unstructured Designs\thanksref{footnoteinfo}} 
                                               
\thanks[footnoteinfo]{This work was partially supported by the US National Science Foundation under grant {\color{black}ECCS 1711004.}}

\author[ncsu]{Jishnudeep Kar}\ead{jkar@ncsu.edu},    
\author[ok]{He Bai}\ead{he.bai@okstate.edu},               
\author[ncsu]{Aranya Chakrabortty}\ead{achakra2@ncsu.edu}  

\address[ncsu]{North Carolina State University, USA}  
\address[ok]{Oklahoma State University, USA}      
          

\begin{abstract}                          
We develop data-driven reinforcement learning (RL) control designs for input-affine nonlinear systems. We use Carleman linearization to express the state-space representation of the nonlinear dynamical model in the Carleman space, and develop a real-time algorithm that can learn nonlinear state-feedback controllers using state and input measurements in the infinite-dimensional Carleman space. Thereafter, we study the practicality of having a finite-order truncation of the control signal, followed by its closed-loop stability analysis. Finally, we develop two additional designs that can learn structured as well as sparse representations of the RL-based nonlinear controller, and provide theoretical conditions for ensuring their closed-loop stability. We present numerical examples to show how our proposed method generates closed-loop responses that are close to the optimal performance of the nonlinear plant. We also compare our designs to other data-driven nonlinear RL control methods such as those based on neural networks, and illustrate their relative advantages and drawbacks. \end{abstract}

\end{frontmatter}
\section{INTRODUCTION}
In recent years, reinforcement learning (RL), proposed originally by Sutton and Barto as an approach for stochastic policy optimization in \cite{sutton}, has emerged as an effective and popular tool for adaptive optimal control of dynamical systems with unknown state-space models \cite{lewis,jiang,ben}. The method has specifically received attention for the design of optimal state-feedback based linear quadratic regulators (LQR) in the absence of a precise knowledge of the plant model, using only measurements of the states and the control inputs over time \cite{ben,powell,lewis,jiang}. A technique known as \textit{policy iteration} has been proposed in these seminal papers, augmented by several other recent advanced optimization methods \cite{fazel}, to solve the model-free LQR design problem using RL. 

{\color{black}The extension of the policy iteration approach or of other alternative techniques for applying RL to nonlinear dynamic models, however, has only been sparsely reported due to the lack of theoretical guarantees on stability and convergence.} A few {\it approximation}-based RL algorithms have been proposed for nonlinear systems, but they often come with stringent limits on scalability and learning times. RL controllers for small-signal linearized models of nonlinear systems have been derived in \cite{sayak} based on singular perturbation approximations, but the controller is only valid under {\it small} perturbations to the nonlinear dynamics, and fails to stabilize the system under any generic large-signal perturbation. Alternative designs have been reported in \cite{persis} on RL for specific nonlinear structures such as polynomials, and in \cite{nn,nn2,warren1}, where neural networks (NN) are used to approximate the nonlinear functions in the Hamilton–Jacobi–Bellman equation for a general nonlinear plant using appropriate activation functions and training weights. These methods guarantee a stabilizing RL controller for a large number of neurons, but as the system size grows, the complexity of the NN increases, thereby leading to longer learning times that may not be tenable for real-time control.  

In this paper, we develop a new theory for policy iteration based RL control for nonlinear systems based on an infinite-series expansion. Various asymptotic approaches for infinite-series representation and corresponding linearization of nonlinear systems have been used in the literature, three most common examples being Carleman linearization \cite{carleman}, Koopman linearization \cite{koopman}, and Lie algebra \cite{lie}. While these methods find relative merits and demerits depending on their use and application, our results in this paper are based on the observation that Carleman linearization is the best among the three for policy iteration based RL due to its bilinear functional form. We first express the nonlinear plant as an equivalent Carleman bilinear system in an infinite-dimensional space. We derive the Lyapunov equation that a quadratic controller must satisfy for optimality, and thereafter develop both on-policy and off-policy algorithms for learning this controller such that it converges to the optimal controller in the infinite-dimensional space starting from a stabilizing control.  

Since the infinite-dimensional control cannot be implemented in a real application, we next consider a truncated approximation of the Carleman bilinear system, and rederive the RL control design for this truncated model. The truncation order can be chosen to be any integer by the designer, depending on the {\it severity} of the nonlinearity of the plant. We mathematically show how the approximation error impacts the closed-loop performance and the stability margins. We list the steps for online implementation of the truncated controller using both on-policy and off-policy methods. Subsequently, we extend the Carleman based RL theory to learning structured and sparse controllers. The structured solution is presented based on an iterative solution of a generalized Riccati like equation in the Carleman space, whereas the sparse controller is derived using an alternating direction method of multipliers (ADMM) based method in tandem with the Carleman policy iterations. Results are validated using two numerical examples, where we show how our proposed method provides solutions close to the  optimal control resulting from the model-based Hamilton-Jacobi-Bellman equation. We also compare our controller to alternative data-driven control using deep learning, and show its relative advantage in terms of faster learning times.\vspace{-0.15cm}

In addition to our previous paper in \cite{ourACC}, the main contributions of this paper are as follows.\vspace{-0.15cm}
\begin{enumerate}
    \item Develop on-policy and off-policy RL algorithms for optimal control of input-affine nonlinear systems using Carleman linearization
    \item Derive conditions for stability and convergence of the designs with {\it truncated} Carleman realization.
    \item Develop two numerical algorithms, one for structured control and one for sparse control, for the Carleman RL based truncated controller using a generalized Riccati equation in the Carleman space. \vspace{-0.15cm}
\end{enumerate}


The paper is organized as follows. Section II formulates the control design. Section III recapitulates the Carleman linearization approach, followed by Section IV that develops  model-free RL design for the Carleman model.  Section V presents the control policy for $\mathbf{N}^{th}$-order truncation of the Carleman model. Section VI summarizes the  online implementation followed by simulation results in Section VII. Section VIII concludes the paper. \vspace{-0.3cm}

\section{Problem Formulation}
Consider a nonlinear system of the form\vspace{-0.15cm}
\begin{equation}
    \dot{z}(t) = f(z(t)) + g(z(t)){w(t)}, \;\;\; {z}(0)={z}_0
    \label{eqn:nonlinear}\vspace{-0.15cm}
\end{equation}
 where, $z \in \mathbb{R}^n$ is the state, $f(z): \mathbb{R}^n \longrightarrow \mathbb{R}^n$ is an unknown nonlinear but {\color{black} analytic function}, $n$ is the order of the plant, which is assumed to be known, $g(z) \in \mathbb{R}^{n\times k}$ is a known input matrix, and ${w} \in \mathbb R^k$ is the control input. Let $(z^{*},w^{*})$ be the steady state values of $(z,\,w)$  such that $\dot{z}=0$. Define $x:= z - z^{*}$ and $u = w-w^{*}$ as the deviations of the state and the input from their respective equilibrium. The goal is to design a controller $u(t)$ that drives  $x(t)$ to zero while minimizing\vspace{-0.15cm}
\begin{equation}
\min _u J = \dfrac{1}{2}  \int_0^\infty (x^T Q_1 x + u^T R u) dt \;,
    \label{eqn:objfn}
\end{equation}
where, $Q_1 \in \mathbb{R}^{n \times n} \geq 0$ and $R \in \mathbb{R}^{k \times k}>0$ are user-defined weight matrices, to improve the closed-loop transient response. Both $x$ and $w$ are assumed to be measured, meaning that $x_0$ and $w_0$ are available from measurements. One way to solve this problem is by approximating \eqref{eqn:nonlinear} by a linear time-invariant system 
\begin{equation}
    \dot{x}(t) = A x(t) + B_1 u(t), \;\; x(0)= z(0)-z_0,\, u(0)=0
  \vspace{-0.1cm}  \label{eqn:lti}
\end{equation}
where $A \in \mathbb{R}^{n \times n}$ is the Jacobian matrix. If $A$ is known, then the optimal control input $u$ minimizing \eqref{eqn:objfn} is given by $u=-Kx$, where $K \in \mathbb{R}^{1 \times n} = R^{-1}B_1^TP$ is the feedback gain and $P \in \mathbb{R}^{n \times n}$ is the solution of the algebraic Riccati equation
\begin{equation}
    A^TP + PA - PB_1R^{-1}B_1^TP + Q_1 = 0.\vspace{-0.1cm}
    \label{eqn:are}
\end{equation}
The above Riccati equation can also be solved in an iterative manner \cite{lewis} as follows.
\begin{multline}
  \text{solve for $P_k$ : } A^TP_{k} + P_kA - P_{k-1}BR^{-1}B^TP_{k-1} = 0 \\
  \text{update $k$ : } k \hookleftarrow k+1.
\end{multline}

This design, however, is valid only for a small-signal approximation \eqref{eqn:lti}, and will not hold if this approximation does not apply on $f(\cdot)$. Therefore, our approach for the rest of this paper is to express the model~\eqref{eqn:nonlinear} as an infinite-dimensional system, provide the designer with the freedom of truncating the order to any desired finite value depending on the {\it severity} of the nonlinearity $f(\cdot)$ (an upper bound for which may be determined by the designer from prior open-loop experiments on the plant), and finally design a model-free optimal controller \eqref{eqn:objfn} that stabilizes the truncated closed-loop model. We use Carleman's infinite-dimensional approximation technique for the first step, as described next.

\section{Carleman Linearization}

\subsection{State-space representation in Carleman space}
Given $f(z_0)=0$, the nonlinear system \eqref{eqn:nonlinear} can be  written as an infinite power-series expansion as
\cite{motee} 
\begin{equation}
    \dot{x} = \sum_{k=1}^{\infty} A_{1k} x^{[k]}  +  \sum_{i=1}^k g_i(x) u_i,
    \label{eqn:power}
\end{equation}
where $x^{[k]} \in \mathbb{R}^{n_k} = \underbrace{x \otimes x \otimes \hdots x}_{k \text{ times}}$, where $n_k$ are the unique elements in the Kronecker product, $A_{1k} \in \mathbb{R}^{n \times n_k}$ can be computed as in \cite{motee}, and $g_i(x)$ is the $i^{th}$ column of the input matrix $g(x)$ corresponding to input $u_i$. Further, expanding $g_i(x)$ as $g_i(x) = (B_{0i} + \sum_{l=1}^{\infty}B_{li}x^{[i]})$, we have
\begin{equation}
     \dot{x} = \sum_{k=1}^{\infty} A_{1k} x^{[k]}  +  \sum_{i=1}^k  (B_{0i} + \sum_{l=1}^{\infty}B_{li}x^{[i]}) u_i,
\end{equation}
Define $\psi_{\mathbf{N}} = col(x \;,\; x^{[2]} \;,\; \hdots \;, x^{[\mathbf{N}]}), \,\,  N \in \mathbb{Z}^{+}$, where $col(.)$ is a column vector by stacking the embraced vectors. Let $\psi:= \psi_N$ when $\mathbf{N} \rightarrow \infty$. Following \cite{motee}, the model \eqref{eqn:power} can then be expressed in an infinite-dimensional state-space representation as \footnotesize
\begin{equation}
    \dot{\mathbf{\psi}} = \mathcal{A} \psi + \sum_{i=1}^k \underbrace{\Big(\begin{bmatrix}
    B_{0i} \\
    \vline \\
    0 \\
    \vline
    \end{bmatrix}+ \begin{bmatrix}
    B_{1i} & B_{2i} & B_{3i} & \hdots \\
    B_{(2,0)i} & B_{(2,1)i} & B_{(2,2)i} & \hdots \\
    0 & B_{(3,0)i} & B_{(3,1)i} & \hdots\\
    \vdots & \vdots & \vdots & \vdots
    \end{bmatrix} \psi\Big)}_{B_{i(\psi)}} u_i
    \label{eqn:infcarl}
\end{equation}
\begin{equation}
    \mathcal{A } = \begin{bmatrix}
    A_{11} & A_{12} & A_{13} & \hdots \\
    0 & A_{22} & A_{23} & \hdots \\
    0 & 0 & A_{33} & \hdots \\
    \vdots & \vdots & \vdots & \ddots
    \end{bmatrix} , 
    \label{eqn:infAinfB}
\end{equation}\normalsize
where $A_{ij} \in \mathbb{R}^{n_i \times n_j}$ and $B_{i0}\in \mathbb{R}^{n_i \times n_{i-1}}$ are given by
\begin{subequations}
\begin{align}
    A_{ij} = \sum_{k=1}^i (\mathbf{I} \otimes \hdots \otimes \underbrace{A_{1j}}_{k^{th} \text{ position}} \otimes \hdots \mathbf{I}),\\
    B_{i0} = \sum_{k=1}^i (\mathbf{I} \otimes \hdots \otimes \underbrace{B_{1}}_{k^{th} \text{ position}} \otimes \hdots \mathbf{I}),
\end{align}
\label{eqn:AijBijset}
\end{subequations}
with $\mathbf{I}$ being the identity matrix of an appropriate dimension. Defining $B_{\psi} = [B_{1(\psi)} \,\, B_{2(\psi)} \, \hdots ]$, we can write \eqref{eqn:infcarl} in a simplified form as 
\begin{equation}
      \dot{\mathbf{\psi}} = \mathcal{A} \psi + B_{\psi} u.
      \label{eqn:simplinfcarl}
\end{equation}

Based on the model \eqref{eqn:infcarl}-\eqref{eqn:AijBijset}, we next derive the Lyapunov equation that an optimal control gain $K$ in the Carleman space must satisfy to stabilize the closed-loop model. 

\subsection{Derivation of Lyapunov equation}
We rewrite the control objective in~\eqref{eqn:objfn} as
\begin{equation}
    J = \dfrac{1}{2} \int_0^\infty (\psi^T Q \psi + u^T R u) \;dt
    \label{eqn:obj}
\end{equation}
where $Q$ is defined as 
\begin{equation}
    Q = \begin{bmatrix}
    Q_{1} & 0\\
    0 & 0
    \end{bmatrix}, \;\; R > 0,
\end{equation}
and $Q_1$ is a positive definite matrix from \eqref{eqn:objfn}. 
The goal is to find a control feedback gain $K$ such that the control input $u=-K\psi$ when actuated in the system \eqref{eqn:nonlinear} minimizes the objective value \eqref{eqn:obj}. First, we need to present the Lyapunov equation for the infinite-dimensional Carleman state space model. Before proving the theorem, first we present an important lemma.
\begin{lem}
For a control input given by $u = -K \psi$, the closed-loop system can be written as
\begin{equation}
    \dot{\psi} = (\mathcal{A} - B_{\psi} K) \psi = \mathcal{A}_{cl}\psi,
    \label{eqn:closedloopAcl}
\end{equation}
where the closed-loop state matrix $\mathcal{A}_{cl}$ is a constant matrix that is independent of $\psi$.
\end{lem}
\begin{proof}
We will present the proof in brief. Note that since $\psi$ is an infinite dimensional vector, one can see that the term $B_{\psi}K\psi$ can be simplified as $B_{\psi}K\psi = \mathcal{K}\psi$, where $\mathcal{K}$ is a constant infinite dimensional matrix. Therefore \eqref{eqn:closedloopAcl} becomes $\dot{\psi} = (\mathcal{A - K})\psi = \mathcal{A}_{cl}\psi$. Detailed matrix formations can be found in \cite{motee}.
\end{proof}

\begin{theorem}
Given that $(\mathcal{A},Q^{1/2})$ is controllable, if $K$ is the optimal controller minimizing the objective function \eqref{eqn:obj} for the system \eqref{eqn:infcarl}, then $K$ must satisfy the Lyapunov-like infinite-dimensional equation for a symmetric matrix $P$ given by\vspace{-0.15cm}
\begin{equation}
\mathcal{A}_{cl}^T P + P \mathcal{A}_{cl} = - (Q + K^T R K),\vspace{-0.15cm}
    \label{eqn:carllyapunov}
\end{equation}
where $\mathcal{A}_{cl} \psi = \mathcal{A} \psi - B_{\psi} K\psi$.
\end{theorem}
\begin{proof}
Following standard optimal control theory, we define the Hamiltonian as 
\begin{equation}
    \mathcal{H} = \dfrac{1}{2} (\psi^TQ\psi + u^T R u) + \Lambda^T (\mathcal{A}\psi + B_{\psi}u),
\end{equation}
where $\Lambda$ is the Lagrange multiplier in the infinite dimension.

For the first optimality condition, we have\vspace{-0.1cm}
\begin{equation}
    \dfrac{\partial \mathcal{H}}{\partial u} = 0 \longrightarrow u = -R^{-1}B_{\psi}^T\Lambda\vspace{-0.15cm}
    \label{eqn:infcontrol}
\end{equation}
If we have $\Lambda = P\psi$, then the control input can be written as $u = -R^{-1}B_{\psi}^T P \psi$. Since $B_{\psi}$ follows from \eqref{eqn:infcarl}, from \cite{motee2} we can write $u$ as \vspace{-0.06cm}
\begin{equation}
    u = -\underbrace{R^{-1} \Big( \begin{bmatrix}
    B_{1} \\ 0
    \end{bmatrix}^T P + W_P^T
    \Big)}_{K} \psi.\vspace{-0.15cm}
\end{equation}
where $B_1 = [B_{01} \, B_{02} \,\hdots \,B_{0k}]$.
As $\psi$ is an infinite-dimensional vector, using quadrization \cite{motee2} we can write
\begin{multline}
    \psi^TP \begin{bmatrix}
    B_{1i} & B_{2i} & B_{3i} & \hdots \\
    B_{(2,0)i} & B_{(2,1)i} & B_{(2,2)i} & \hdots \\
    0 & B_{(3,0)i} & B_{(3,1)i} & \hdots
    \end{bmatrix} \psi = \psi^T W^i_P \\ \text{and }\,\, W_p = [W^1_P \, W^2_P \, \hdots \, W^k_P]
\end{multline}
 where $W_P$ is matrix of  appropriate dimensions (see details in \cite{motee2}). This is possible because the higher-order terms get absorbed in $\psi$. Using the second optimality condition, we have\vspace{-0.1cm}
\begin{equation}
    \dot{\Lambda} = -\dfrac{\partial \mathcal{H}}{\partial \psi} =  -(Q\psi + K^TRK \psi+ \mathcal{A}_{cl}^T P \psi),\vspace{-0.1cm}
    \label{eqn:intlyap}
\end{equation}
where $\mathcal{A}_{cl} \psi = \mathcal{A} \psi - B_{\psi} K\psi$. Substituting $\Lambda = P \psi$, we have $\dot{\Lambda} = P \dot{\psi}$. Since $\dot{\psi} = \mathcal{A}_{cl} \psi$, and \eqref{eqn:intlyap} must hold for all $\psi$,  we finally get  the Lyapunov-like equation in the infinite-dimensional Carleman space as in \eqref{eqn:carllyapunov}.
\end{proof}
We next discuss an RL control method for the infinite-dimensional system~\eqref{eqn:infcarl}.

\section{RL Control for Lifted System}
\subsection{Policy iteration for Carleman systems}
For the infinite-dimension model \eqref{eqn:infcarl}, the cost-to-go function can be written as \cite{motee2}
\begin{equation}
    V(t) = \int_t^{\infty} (\psi^T Q \psi + u^T R u) \;d\tau = \psi(t)^T P \psi(t),
\end{equation}
where $P$ is the solution of the Lyapunov equation \eqref{eqn:carllyapunov}. Due to the structure of $Q$, we can rewrite $V(t)$ as 
\begin{equation}
     V(t) = \int_t^{\infty} (x^T Q_{1} x + u^T R u) \;d\tau.
     \label{eqn:costtogo}
\end{equation}
For an iteration index $i$, and $P_i$, note that we can write $V_i(t) - V_i(t+T)$ for $T>0$ as
\begin{multline}
    V_i(t) - V_i(t+T)=\psi^T(t)P_i\psi(t) - \psi^T(t+T)P_i\psi(t+T)\\ 
     =\underbrace{\int_t^{t+T} (x^T Q_1 x + \psi^T K_i^TRK_i \psi) d\tau}_{d(t,K_i)} ,
    \label{eqn:diffPsi}
\end{multline}
where $T$ is the time-step for updating $P_i$. We define $\bar{p}_i$, the vectorized version of $P_i$ such that
\begin{equation}
   \psi^T(t)P_i\psi(t) = \bar{p}^T_i \psi(t),
\end{equation}
since $\psi$ is an infinite-dimensional vector.  Therefore, we can write \eqref{eqn:diffPsi} as \begin{equation}
    \bar{p}^T_i \big( \psi(t) - \psi(t+T) \big) = d(t,K_i).
    \label{eqn:pi}
\end{equation}
The solution $P_i$ from \eqref{eqn:pi} yields the new control input 
\begin{equation}
    u = -R^{-1} \Big( \begin{bmatrix}
    B_1 \\ 0
    \end{bmatrix}^T P_i + W_{P_i}^T
    \Big) \psi = -K_{i+1} \psi.
    \label{eqn:learnK}
\end{equation}
Note that the method  presented above is an {\it on-policy} method, meaning that the learned control gain $K_{i}$ is iteratively implemented in the plant to learn the next control gain $K_{i+1}$. Such learning will generally require an exploration noise to persistently excite the system to learn $P_i$ correctly. The online implementation of the above steps with exploration noise will be discussed in Section VI. Next, we prove that the policy iteration method presented above leads to the stabilizing optimal control. For this, it is sufficient to establish that
\begin{enumerate}
    \item Solving \eqref{eqn:pi} is equivalent to solving the Lyapunov equation \eqref{eqn:carllyapunov}.
    \item Starting from a stabilizing $P_0$, all $P_i, \; i\geq 1 $ are stabilizing.
    \item For $i \longrightarrow \infty$, $P_i \longrightarrow P^{*}$, where $P^{*}$ is the optimal cost matrix for \eqref{eqn:infcarl}.
\end{enumerate}

\subsection{Theoretical Results}
\begin{theorem}
Solving the data-driven equation \eqref{eqn:pi} is equivalent to solving the Lyapunov equation $\mathcal{A}_{cl,i-1}^T P_i + P_i \mathcal{A}_{cl,i-1} = - (Q + K_i^T R K_i)$, where $\mathcal{A}_{cl,i-1}\psi(t) = (\mathcal{A} - B(\psi)R^{-1}B(\psi)^TP_{i-1})\psi(t)$. 
\end{theorem}
\begin{proof}
Consider the Lyapunov function of the system $ V_i(t) = \psi(t)^TP_i\psi(t)$. Then we can write 
\begin{gather}
    \dfrac{d \big(\psi(t)^TP_i\psi(t) \big)}{dt} = \psi(t)^T (\mathcal{A}_{cl,i-1}^TP_i + P_i\mathcal{A}_{cl,i-1})\psi(t),\vspace{-0.2cm}
    \label{eqn:psiPpsi}
\end{gather}
Since we want to find $P_i > 0$ such that the following Lyapunov equation holds\vspace{-0.2cm}
\begin{equation}
    \mathcal{A}_{cl,i-1}^TP_i + P_i\mathcal{A}_{cl,i-1} = - (Q + K_i^TRK_i),
    \label{eqn:lyapiter}
\end{equation}
we integrate \eqref{eqn:psiPpsi} from $t$ to $t+T$ and substitute~\eqref{eqn:lyapiter} to obtain
\begin{multline}
     \psi(t+T)^TP_i\psi(t+T)- \psi(t)^TP_i\psi(t) \\ = \int_t^{t+T} -(\psi(\tau)^T \big( K_i^TRK_i + Q\big)\psi(\tau) d\tau.\vspace{-0.2cm}
     \label{eqn:psiPiter}
\end{multline}\vspace{-0.3cm}
\end{proof}

\begin{theorem}
Starting from a stabilizing $(P_0,K_0)$, if $P_i$ is updated according to \eqref{eqn:diffPsi}, then the closed-loop system $\mathcal{A}_{cl,i}, \forall i \geq 1$ is stable.
\end{theorem}

\begin{proof}
Consider the Lyapunov function\vspace{-0.1cm}
\begin{equation}
    V_i(t) = \psi(t)^TP_i\psi(t).\vspace{-0.12cm}
\end{equation}
Taking derivative on both sides, we have\vspace{-0.2cm}
\begin{multline}
    \dot{V}_i(t) = \psi(t)^T \Bigg( P_i \big(\mathcal{A} - B_{\psi} R^{-1}B_{\psi}^TP_i\big) \\ + \big(\mathcal{A} - B_{\psi} R^{-1}B_{\psi}^TP_i\big)^T P_i \Bigg)\psi(t),\vspace{-0.3cm}
\end{multline}
which can be rewritten as\vspace{-0.2cm}
\begin{multline}
    \dot{V}_i(t) = \psi(t)^T \Bigg( P_i \big(\mathcal{A} - B_{\psi} R^{-1}B_{\psi}^TP_{i-1}\big)\\ + \big(\mathcal{A} - B_{\psi} R^{-1}B_{\psi}^TP_{i-1}\big)^T P_i \Bigg)\psi(t) \\
    + \psi(t)^T \Bigg( P_i \big(  B_{\psi} R^{-1}B_{\psi}^T (P_{i-1} - P_{i}) \big)  \\ +   \big( (P_{i-1} - P_{i}) B_{\psi} R^{-1}B_{\psi}^T  \big)P_i \Bigg) \psi(t)
\end{multline}
Recalling from \eqref{eqn:lyapiter} that
\begin{equation}
    \mathcal{A}_{cl,i-1}^T P_i + P_i\mathcal{A}_{cl,i-1} = - \big( K_i^TRK_i + Q),
\end{equation}
$\dot{V}_i(t)$ can be written as 
\begin{multline}
\dot{V}_i(t) =   -\psi(t)^T \big(K_i^TRK_i + Q \big) \psi(t) \\ +   \psi(t)^T \Bigg( P_i \big(  B_{\psi} R^{-1}B_{\psi}^T (P_{i-1} - P_{i}) \big) \\ +   \big( (P_{i-1} - P_{i}) B_{\psi} R^{-1}B_{\psi}^T  \big)P_i \Bigg) \psi(t).
\label{eqn:Vi3}
\end{multline}
We also have $ K_i \psi(t) = R^{-1}B(\psi)^T P_{i-1} \psi(t)$.
Therefore, we can write  \eqref{eqn:Vi3} as 
\begin{multline}
    \dot{V}_i(t) = -\psi(t)^T \Bigg( (P_{i-1} - P_{i}) B_{\psi} R^{-1}B_{\psi}^T (P_{i-1} - P_{i}) \\  + Q +  P_{i} B_{\psi} R^{-1}B_{\psi}^T P_{i}\Bigg) \psi(t).
    \label{eqn:Vi4}
\end{multline}
From \eqref{eqn:Vi4}, we can see that $\dot{V}_i \leq 0 \,, \forall \psi$. Setting $\dot V_i=0$ leads to $x=0$. 
Since we also have 
\begin{equation}
    \lim_{x \rightarrow 0} x^{[k]} \rightarrow 0 \, , \forall k \in \mathbb{Z}^{+} \Leftrightarrow \lim_{x \rightarrow 0} \psi \rightarrow 0,
\end{equation}
we conclude that $\dot V=0$ implies $\psi=0$.
Therefore, it follows from LaSalle's invariance principle~\cite{lassale} that $\psi = 0$ is asymptotically stable. In other words, if we start from a stabilizing $P_0$, and use \eqref{eqn:diffPsi} to update $P_i$, all $P_i$ and corresponding $K_i$ will be stabilizing. 
\end{proof}

\subsubsection{Convergence}
To prove this, we follow the steps listed in Section 4 of the Aganov-Gajic algorithm \cite{gajic} for a general bilinear system. First, we present the following lemma  \ref{theorem:convergence} to prove convergence of $P_i$.

\begin{lemma}
The sequence $\psi_i$ uniformly converges to $\psi^{*}$ and $u_i$ uniformly converges to $u^{*}$ for every $t \in [0,T] \; , 0< T <\infty $ given the iterative updates are made as \begin{enumerate}
    \item $P_i$ is updated as solution of $\psi^T \big(P_i \mathcal{A} + \mathcal{A}^TP_i + Q - P_iB_{\psi_{i-1}}R^{-1}B_{\psi_{i-1}}^TP_i \big) \psi = 0$,
    \item Control input is $u_i (t) = -R^{-1}B_{\psi}P_i \psi(t)$,
\end{enumerate}
where $\psi_{i-1}$ is the trajectory due to $P_{i-1}$.
\label{theorem:convergence}
\end{lemma}
\begin{proof}
Proof follows from \cite{gajic}, by replacing $x$ by $\psi$, and having $J= P \psi$, where $P$ is independent of $\psi$ (as proved in Theorem 1) in the HJB equation (3) of \cite{gajic}.
\end{proof}

\begin{theorem}
The sequence $P_i$ uniformly converges to $P^{*}$ for every $t \in [0,T] \; , 0< T <\infty $ given the iterative updates for $P_i$ and $u_i$ are made according to (21), (22).
\end{theorem}

\begin{proof}
In the Aganov-Gajic algorithm (Algorithm 3.2 in \cite{gajic}) and \cite{bilinearRL}, it is shown that solving steps 1 and 2 of Lemma \ref{theorem:convergence} is equivalent to the iterative solution of
\begin{equation}
    \mathcal{A}_{i-1}^TP_i + P_i \mathcal{A}_{i-1} = -(Q + P_{i-1}B_{\psi_{i-1}}R^{-1}B_{\psi_{i-1}}^TP_{i-1}),
    \label{eqn:Ahati-1}
\end{equation}
where $\mathcal{A}_{i-1} \psi_{i-1} = (\mathcal{A} - B_{\psi_{i-1}}R^{-1}B_{\psi_{i-1}}^TP_{i-1})\psi_{i-1}$. Pre- and post-multiplying \eqref{eqn:Ahati-1} by $\psi_{i-1}^T$ and $\psi_{i-1}$, we obtain
\begin{multline}
    \psi_{i-1}(t)^T(\mathcal{A}_{i-1}^TP_i + P_i \mathcal{A}_{i-1})\psi_{i-1} = \\ -\psi_{i-1}^T(Q + P_{i-1}B_{\psi_{i-1}}R^{-1}B_{\psi_{i-1}}^TP_{i-1})\psi_{i-1}(t).\vspace{-0.25cm}
    \label{eqn:psiAhati-1}
\end{multline}
From Section I, we know that $\hat{A}_{i-1}\psi_{i-1}(t) = A_{cl,i-1}\psi_{i-1}(t)$ and $R^{-1}B_{\psi_{i-1}}^TP_{i-1}\psi_{i-1}(t) = K_i \psi_{i-1}(t)$. 
Therefore,  \eqref{eqn:psiAhati-1} can be written as 
\begin{multline}
    \psi_{i-1}(t)^T(\mathcal{A}_{cl,i-1}^TP_i + P_i \mathcal{A}_{cl,i-1})\psi_{i-1}(t) = \\ - \psi_{i-1}(t)^T(Q + K_i^TRK_i)\psi_{i-1}(t)^T.
    \label{eqn:proof2}
\end{multline}
Note that in \eqref{eqn:psiPiter}, $\psi(t)$ is the trajectory due to $P_{i-1}$. Therefore $\psi(t) = \psi_{i-1}(t)$. From Lemma \ref{theorem:convergence} and the equivalence of \eqref{eqn:proof2} and (19)-(22) to the model-based results  in \cite{gajic}, a similar proof of convergence follows for $P_i$ \cite{gajic}.
\end{proof}
\noindent It is also shown in  \cite{gajic} that for a finite large iteration index $i$, if $||P_i - P^{*}|| < \gamma$, where $\gamma$ is sufficiently small, the deviation from the optimal trajectory is $\psi^{*}(t) = \psi_i(t) + \mathcal{O}(\gamma)$. Since we have shown that all $P_i$ are stabilizing, we can write $J_i = \psi_0^T P_i \psi_0 < \infty$ and obtain\vspace{-0.12cm}
 \begin{equation}
     J_{i} - J^{*} = \psi_0^T(P_i - P^{*})\psi_0.\vspace{-0.1cm}
 \end{equation}
 Taking norm and writing $||J|| = J$, we further obtain\vspace{-0.1cm}
 \begin{equation}
    J^{*} \leq J_{i} \leq J^{*} + ||\psi_0||^2 \gamma. \vspace{-0.1cm}
 \end{equation}
 If each state in the initial perturbation satisfies  $||x_i|| < 1, i=1,2,..,n$, then $||\psi_0||$ is bounded. 
In \cite{bilinearRL}, it is also shown that for a bilinear system, the sequence $\{P_i\}$ is monotonically decreasing sequence, i.e
\begin{equation}
    P_0 \geq P_1 \geq \hdots \geq P^{*}.
\end{equation}

\section{$\mathbf{N}^{th}$ order truncation}
In practice, it would not be possible to construct an infinite-dimensional vector to learn the controller \eqref{eqn:learnK}. Therefore, one must resort to using a finite truncation of the lifted system, and executing the control design in the truncated space. A finite $\mathbf{N}^{th}$ order truncation of $\psi$ can be written as \vspace{-0.25cm}
\begin{equation}
\psi_\mathbf{N} = col(x \,,\, x\otimes x \,,\, \hdots \,,\, \underbrace{x \otimes \hdots \otimes x}_{\mathbf{N} \text{ times}}) \in \mathbb{R}^{\sum\limits_{i=1}^n n_i}.\vspace{-0.15cm}
\label{eqn:psiNdefn}
\end{equation}
The Carleman linearized dynamics for a truncated system can be written as \vspace{-0.25cm}
\begin{multline}
        \dot{\psi}_{\mathbf{N}}
 = \mathcal{A}_{\mathbf{N}} \psi_{\mathbf{N}} + \\ \sum_{i=1}^k \Big(\underbrace{\begin{bmatrix}B_{0i} \\ 0 \end{bmatrix}+ \begin{bmatrix}
    B_{1i} & B_{2i} & \hdots & B_{\mathbf{N}i} \\
    B_{(2,0)i} & B_{(2,1)i} & B_{(2,2)i} & \hdots \\
    0 & B_{(3,0)i} & B_{(3,1)i} & \hdots
    \end{bmatrix} \psi_{\mathbf{N}}\Big)}_{B_{i(\psi_{\mathbf{N}})}} u_i
    \label{eqn:Ncarl}
\end{multline}
where $\mathcal{A}_{\mathbf{N}}$ is the truncated state matrix, $B_1$ is the same as in (1), and $B_{2_{\mathbf{N-1}}}$ is the truncated version of $B_2$ in \eqref{eqn:infAinfB}. We define $B_{\psi_{\mathbf{N}}} = \begin{bmatrix} B_{1(\psi_{\mathbf{N}}} & \hdots & B_{k(\psi_{\mathbf{N}})} \end{bmatrix}$.

\subsection{LQR control objective and policy iteration}

The LQR objective function for the bilinear model is 
\begin{equation}
    J = \dfrac{1}{2} \int_0^{\infty}(\psi_{\mathbf{N}}
^T Q_{\mathbf{N}}  \psi_{\mathbf{N}} + u^T R u) \;dt
    \label{eqn:carlobj}
\end{equation}
where $Q_{\mathbf{N}}$ is defined as \vspace{-0.15cm}
\begin{equation}
    Q_{\mathbf{N}} = \begin{bmatrix}
    Q_{1} & 0\\
    0 & 0
    \end{bmatrix}.
\end{equation}
A truncated version of $P$ can be written in the compartmental form as \vspace{-0.2cm}
\begin{equation}
    P_{\mathbf{N}} = \begin{bmatrix}
    P_{\mathbf{N},11} & P_{\mathbf{N},12}\\
    P_{\mathbf{N},21} & P_{\mathbf{N},22}\\
    \end{bmatrix},\vspace{-0.15cm}
\end{equation}
where $P_{\mathbf{N},11} \in \mathbb{R}^{n \times n}$, $P_{\mathbf{N},12} \in \mathbb{R}^{n \times \text{dim}(\psi_{\mathbf{1-N}})}$, and $P_{\mathbf{N},21} \in \mathbb{R}^{\text{dim}(\psi_{\mathbf{N-1}}) \times n}$. Next, we truncate the control input $u$ such that it can be expressed as a linear feedback for the truncated Carleman system $u = -K_{\mathbf{N}} \psi_{\mathbf{N}}$. 
Similar to the infinite-dimensional case, the control input would be $u = -R^{-1}B_{\psi_{\mathbf{N-1}}}^T P_{\mathbf{N}}\psi_{\mathbf{N}}$. However, this $u$ has extra  terms with orders higher than $\mathbf{N}$ that are not linear in $\psi_{\mathbf{N}}$. Therefore, truncating these higher-order terms to represent $u$ as a linear feedback of the truncated Carleman states, we have\vspace{-0.15cm}
\begin{multline}
    u_i = -\dfrac{1}{r_{ii}}\Bigg([B_{0i}^T \,\,\, 0] P_{\mathbf{N}}\psi_{\mathbf{N}} + \\ \sum\limits_{k=1}^{\mathbf{N}} \sum\limits_{j+m \leq \mathbf{N}} x^{[j]} B_{(k,j-1)i} P_{\mathbf{N},km} x^{[m]}\Bigg)
\end{multline}
where 
\begin{equation}
    P_{\mathbf{N}} = \begin{bmatrix}
    P_{\mathbf{N},11} & P_{\mathbf{N},12} & \hdots & P_{\mathbf{N},1\mathbf{N}}\\
    \vdots & \vdots &\hdots &\vdots \\
     P_{\mathbf{N},\mathbf{N}1} & P_{\mathbf{N},\mathbf{N}2} & \hdots & P_{\mathbf{N},\mathbf{N}\mathbf{N}}\\
    \end{bmatrix} , 
    \label{eqn:trunccontrolinput}
\end{equation}
and $P_{\mathbf{N},km} \in \mathbb{R}^{\text{dim}(x^{[k]}) \times \text{dim}(x^{[m]})}$. Note, that in \eqref{eqn:trunccontrolinput}, the control input is truncated because we neglect all the higher-order terms for $j+m > \mathbf{N}$. We will denote the trucation error as 
\begin{equation}
    \epsilon_{trunc}(P_{\mathbf{N}}) =  \dfrac{1}{r_{ii}}\Bigg(\sum\limits_{k=1}^{\mathbf{N}} \sum\limits_{j+m > \mathbf{N}}^{j+m \leq 2\mathbf{N}} x^{[j]} B_{(k,j-1)i} P_{\mathbf{N},km} x^{[m]}\Bigg)
    \label{eqn:truncexp}
\end{equation}

\subsection{Reinforcement learning for the truncated Carleman model}
Similar to \eqref{eqn:diffPsi}, by replacing $\psi$ by $\psi_{\mathbf{N}}$, the update can be written as \vspace{-0.25cm}
\begin{multline}
    \psi_{\mathbf{N}}(t)^T(t)P_{\mathbf{N}(i)}\psi_{\mathbf{N}}(t) -  \psi_{\mathbf{N}}(t)^T(t+T)P_{\mathbf{N}(i)}\psi_{\mathbf{N}}(t+T) = \\ \int_t^{t+T} (x^T Q_1 x + \psi_{\mathbf{N}}^TK_{\mathbf{N}(i)}^TRK_{\mathbf{N}(i)})\psi_{\mathbf{N}} dt.
    \label{eqn:diffPsifinite}
\end{multline}
The LHS is vectorized such that it can be written as $\bar{p}_{\mathbf{N}(i)} \big( \bar{\psi}_{\mathbf{N}}(t) - \bar{\psi}_{\mathbf{N}}(t+T) \big)$, where $\psi_{\mathbf{N}}(t)^T(t)P_{\mathbf{N}(i)}\psi_{\mathbf{N}}(t) = \bar{p}_{\mathbf{N}(i)}\bar{\psi}_{\mathbf{N}}(t)$. Solving \eqref{eqn:diffPsifinite} using real-time measurements of the states and control input (as will be discussed in Section VI), one can obtain $P_{\mathbf{N}(i)}$. Then, the new updated control input $u_z, z=1,..k$, which is linear in $\psi_{\mathbf{N}}$, is computed as 
\begin{multline}
    u_z = -\dfrac{1}{r_{zz}}\Bigg([B_{0z}^T \,\,\, 0] P_{\mathbf{N}}^{(i)}\psi_{\mathbf{N}} + \\ \sum\limits_{k=1}^{\mathbf{N}} \sum\limits_{j+m \leq \mathbf{N}} x^{[j]} B_{(k,j-1)z} P^{(i)}_{\mathbf{N},km} x^{[m]}\Bigg).
    \end{multline}\normalsize
It is important to note that the control input is being forcibly truncated to make the control input linear in $\psi_{\mathbf{N}}$ and independent of the higher-order terms. Next, we discuss the impact of this truncation, and derive sufficient conditions to ensure closed-loop stability. 

\subsection{Stability analysis for truncated approximation}\label{subsec:trunstab}
First, we state Lemma \ref{lemma:stability} for an ideal closed-loop truncated system. Note that an ideal closed-loop truncation is such that the closed-loop system dynamics is linear in $\psi_{\mathbf{N}}$. For example, for $\mathbf{N}=2$, the truncated control input is $u_i=-K_{11,i}\psi^{[1]} - K_{12,i}\psi^{[2]}$, then the ideal truncated closed-loop dynamics is given by
\begin{equation}
    \begin{bmatrix}
    \dot{\psi}^{[1]}\\\dot{\psi}^{[2]}
   \end{bmatrix} = \begin{bmatrix}
    A_{11,cl} & A_{12,cl} \\ 0 & A_{21,cl}
    \end{bmatrix}\begin{bmatrix}
   \psi^{[1]}\\\psi^{[2]}
    \end{bmatrix},
    \label{eqn:idealtruncsys}
\end{equation}
where $A_{11,cl} = A_{11} - \sum_{i=1}^m B_{0i}K_{11,i}$, $A_{12,cl} = A_{12} - \sum_{i=1}^m (B_{0i}K_{12,i} + B_{1,i}\otimes K_{11,i})  $, and $A_{21,cl} = A_{21} - \sum_{i=1}^mB_{(2,0)i}\otimes K_{11,i} $.
\begin{lemma}\label{lemma:stability}
Consider the ideal system $\dot\psi_{\mathbf{N}}= \mathcal{A}_{cl} \psi_{\mathbf{N}}$, where $\psi_{\mathbf{N}}=[(\psi^{[1]})^T,\cdots,(\psi^{[\mathbf{N}])^T}]^T$. Suppose that $\psi^{[1]}(t)\rightarrow 0$ as $t\rightarrow\infty$ for any $\psi^{[k]}(0)$, $k=1,\cdots,\mathbf{N}$. Then $\psi^{[k]}(t)\rightarrow 0, \, \forall k=2,3,..,\mathbf{N}$.
\end{lemma}
\begin{proof}
Consider the truncated dynamics for $\mathbf{N}=2$ given by
 \eqref{eqn:idealtruncsys}. Since $\psi^{[1]}(t)\rightarrow 0$ for any $\psi^{[k]}(0)$, $k=1,2$, we choose $\psi^{[2]}(0)=0$. Then $\psi^{[1]}(t)\rightarrow 0$, $\forall \psi^{[1]}(0)$, implies that $A_{11,cl}$ must be Hurwitz. Since $A_{21,cl} = A_{11,cl} \otimes I + I \otimes A_{11,cl}$, $A_{21,cl}$ is also Hurwitz, which means that $\psi^{[2]}(t)\rightarrow 0$,  $\forall \psi^{[2]}(0)$. Similar proof follows for $\mathbf{N}>2$.
\end{proof}

We next present a lemma, which is a truncated version of Theorem 3, that holds under certain conditions on $Q_1$.
\begin{lemma}\label{lem:lemma8}
Starting from a stabilizing $(P_0,K_0)$, for a sufficiently large $\lambda_{\text{min}}(Q_1)$, if $P_{\mathbf{N}(i)}$ is updated according to \eqref{eqn:diffPsifinite}, then the closed-loop system $\mathcal{A}_{cl,i}, \forall i \geq 1$ is stable.
\end{lemma}
\begin{proof}
Consider the Lyapunov function
\begin{equation}
    V_i(t) =  \psi_{\mathbf{N}}(t)^TP_{\mathbf{N}}^i \psi_{\mathbf{N}}(t).
\end{equation}
Differentiating with respect to time $t$, we have
\begin{multline}
   \dot{V}_i(t) =  \psi_{\mathbf{N}}^T \Bigg( P_{\mathbf{N}(i)} (\mathcal{A}_{\mathbf{N}} - B_{\psi_{\mathbf{N}}} R^{-1}B_{\psi_{\mathbf{N}}}^T P_{\mathbf{N}(i)}) \\ + (\mathcal{A}_{\mathbf{N}} - B_{\psi_{\mathbf{N}}} R^{-1}B_{\psi_{\mathbf{N}}}^T P_{\mathbf{N}(i)})^T  P_{\mathbf{N}(i)}\Bigg) \psi_{\mathbf{N}},
\end{multline}
\normalsize 
which is further rewritten as
\begin{multline}
    \dot{V}_i(t) =  \psi_{\mathbf{N}}^T \Bigg( P_{\mathbf{N}(i)} (\mathcal{A}_{\mathbf{N}} - B_{\psi_{\mathbf{N}}} R^{-1}B_{\psi_{\mathbf{N}}}^TP_{\mathbf{N}(i-1)} \\ + (\mathcal{A}_{\mathbf{N}} - B_{\psi_{\mathbf{N}-1}} R^{-1}B_{\psi_{\mathbf{N}}}^T{P_{\mathbf{N}(i-1)}})^T {P_\mathbf{N}(i)} \Bigg) \psi_{\mathbf{N}}(t)\\
    +  \psi_{\mathbf{N}}(t)^T \Bigg( P_{\mathbf{N}(i)} \big(  B_{\psi_{\mathbf{N}}} R^{-1}B_{\psi_{\mathbf{N}}}^T (P_{\mathbf{N}(i-1)} - P_{\mathbf{N}(i)}) \big) \\ +   \big( (P_{\mathbf{N}(i-1)} - P_{\mathbf{N}(i)}) B_{\psi_{\mathbf{N}}} R^{-1}B_{\psi_{\mathbf{N}}}^T  \big)P_{\mathbf{N}(i)} \Bigg)  \psi_{\mathbf{N}}(t).
\end{multline}\normalsize
 Also,  we have
\begin{equation}
    R^{-1}\begin{bmatrix}
    B_1\\B_{2}  \psi_{\mathbf{N-1}}
    \end{bmatrix}^TP_{\mathbf{N}(i)} \psi_{\mathbf{N}} = K_{i+1}  \psi_{\mathbf{N}}  + \epsilon_{trunc}(P_{\mathbf{N}(i)}) .
\end{equation}
Recalling from \cite{motee} that
\begin{equation}
    \mathcal{A}_{\mathbf{N},cl}^T P_{\mathbf{N}(i)} + P_{\mathbf{N}(i)} \mathcal{A}_{\mathbf{N}cl} = - \big( K_{\mathbf{N}(i)}^TRK_{\mathbf{N}(i)} + Q),
\end{equation}
and  substituting \vspace{-0.1cm}
\begin{multline}
   K_{\mathbf{N}(i)}\psi_{\mathbf{N}}(t)  =  R^{-1}B_{\psi_{\mathbf{N}}}^TP_{\mathbf{N},(i-1)} \psi_{\mathbf{N}}(t)  - \epsilon_{trunc}(P_{\mathbf{N}(i-1)}),
\end{multline}
 we have
\begin{multline}
    \dot{V}_i(t) = -\psi_{\mathbf{N}}(t)^T \Bigg( (P_{\mathbf{N}(i-1)} - P_{\mathbf{N}(i)}) B_{\psi_{\mathbf{N}}} R^{-1}\\B_{\psi_{\mathbf{N}}}^T (P_{\mathbf{N}(i-1)} - P_{\mathbf{N}(i)})  + Q \\+  P_{\mathbf{N}(i)} B_{\psi_{\mathbf{N}}} R^{-1}B_{\psi_{\mathbf{N}}}^T P_{\mathbf{N}(i)}\Bigg) \psi_{\mathbf{N}}(t) + \Delta_i(\psi_{\mathbf{N}}),
\end{multline}
\normalsize
where $\Delta_i (\psi_{\mathbf{N}})$ consists of the higher-order terms as:
\begin{equation}
    \Delta_i(\psi_{\mathbf{N}}) = 2\psi_{\mathbf{N}}^T (P_{\mathbf{N}(i-1)} -  P_{\mathbf{N}(i)})B_{\psi_{\mathbf{N}-1}}\epsilon_{trunc}(P_{\mathbf{N}(i-1)}).
\end{equation}\normalsize

Consider the ideal case, where the perturbation term due to the truncation $\Delta_i(\psi_{\mathbf{N}}) = 0$. Since $Q_{\mathbf{N}}$ \eqref{eqn:carlobj} is designed such that $\psi_{\mathbf{N}}^T Q_{\mathbf{N}}\psi_{\mathbf{N}} = (\psi^{[1]})^T Q_1 \psi^{[1]}$, it follows that $\dot{V}_i \leq 0$ and thus $\psi^{[1]}(t)\rightarrow 0$. From Lemma \ref{lemma:stability}, we further conclude that for the unperturbed system, $\psi_{\mathbf{N}}(t)\rightarrow 0$.
Hence, $\psi_{\mathbf{N}}=0$ is asymptotically stable. It then follows from standard control theory that for a given positive definite matrix $\Gamma$, there exists a Lyapunov function $\bar{V}_i = \psi_{\mathbf{N}}^T \bar{P}_{\mathbf{N}}^i \psi_{\mathbf{N}}$ for the unperturbed system  such that $\dot{\bar{V}}_i \leq -\psi_{\mathbf{N}}^T \Gamma \psi_{\mathbf{N}}<0$. 


Next, we calculate derivative of $\bar{V}_i$ for the perturbed case $\Delta_i \neq 0$. We obtain $\dot{\bar{V}}_i$ as
\begin{equation}
    \dot{\bar{V}}_i = \psi_{\mathbf{N}}^T \bar{P}_{\mathbf{N}}^i \big(\mathcal{A}_{cl} \psi_{\mathbf{N}} - \delta (\psi_{\mathbf{N}}) \big) +  \big(\mathcal{A}_{cl} \psi_{\mathbf{N}} - \delta (\psi_{\mathbf{N}}) \big)^T \bar{P}_{\mathbf{N}}^i \psi_{\mathbf{N}},
    \label{eqn:Vibar}
\end{equation}
where $\delta(\psi_{\mathbf{N}}) = B(\psi_{\mathbf{N}})\epsilon_{trunc}(P_{\mathbf{N}}^i)$. We further get
\begin{equation}
\begin{split}
    \dot{\bar{V}}_i & = \psi_{\mathbf{N}}^T(\bar{P}_{\mathbf{N}}^i \mathcal{A}_{cl} + \mathcal{A}^T_{cl}\bar{P}_{\mathbf{N}}^i) \psi_{\mathbf{N}}  - 2 \psi_{\mathbf{N}}^T\bar{P}_{\mathbf{N}}^i\delta(\psi_{\mathbf{N}})\\
   & \leq -\psi_{\mathbf{N}}^T \Gamma \psi_{\mathbf{N}} - 2 \psi_{\mathbf{N}}^T\bar{P}_{\mathbf{N}}^i \delta(\psi_{\mathbf{N}}).
    \end{split}
    \label{eqn:trunV2}
\end{equation}
Note that the perturbation term $\mathcal{O}(2\psi_{\mathbf{N}}^T\bar{P}_{\mathbf{N}}^i\delta(\psi_{\mathbf{N}})) > $ is higher-order in $\psi_{\mathbf{N}}$ than the $\mathcal{O}(\psi_{\mathbf{N}}^2)$ term. Therefore, for any $\gamma >0$, there exists $d_1>0$, such that $\|2\psi_{\mathbf{N}}^T\bar{P}_{\mathbf{N}}^i\delta(\psi_{\mathbf{N}})\|< \gamma \|\psi_{\mathbf{N}}\|^2$, $\forall \|\psi_{\mathbf{N}}\|<d_1$. It then follows that $\forall \|\psi_{\mathbf{N}}\|<d_1$,
\begin{equation}
    \dot {\bar V}_i \leq -\psi_{\mathbf{N}}^T \Gamma \psi_{\mathbf{N}} +2\gamma \|\psi_{\mathbf{N}}\|^2\leq -\psi_{\mathbf{N}}^T (\lambda_{min}(\Gamma) - 2\gamma) \psi_{\mathbf{N}}.
    \label{eqn:truncerror}
\end{equation}
Choosing $\gamma < 1/2\lambda_{min}(\Gamma)$ ensures $\dot {\bar V}_i$ is negative definite. We conclude that $\psi_{\mathbf{N}}=0$ is asymptotically stable.\end{proof}


Note that this $d$ depends on the value of $Q_{\mathbf{N}}$ and $R$. Particularly, from \eqref{eqn:trunV2}, we can see that for a large enough $Q_{\mathbf{N}}$ and $R$, we can achieve a larger $d$. In our simulation studies we see that a small enough initial condition and a large enough $Q_{\mathbf{N}}$ guarantees stability. We compare the stability and convergence for decreasing $Q_{\mathbf{N}}$ in Sec. \ref{sec:results}.

\subsection{Stability analysis for infinite-dimensional system}
Lemma~\ref{lem:lemma8} applies to the truncated system~\eqref{eqn:Ncarl} with the learned control $u$ also truncated to the $\mathbf{N}$th order. However, the original system~\eqref{eqn:infcarl} is an infinite-dimensional system. Therefore, there is a discrepancy between the trajectories of~\eqref{eqn:Ncarl} and~\eqref{eqn:infcarl}. In this section, we show that the $\mathbf{N}$th-order control $u$ learned from the trajectory of~\eqref{eqn:infcarl} is still asymptotically stabilizing, provided that $x_0$ is small. Specifically, we  establish that the discrepancy between~\eqref{eqn:Ncarl} and~\eqref{eqn:infcarl} is small if $\mathbf{N}$ is chosen sufficiently large and $x_0$ is sufficiently small. Then the $\mathbf{N}$th-order control $u$ learned based on~\eqref{eqn:infcarl} still stabilizes~\eqref{eqn:Ncarl}. Finally, in Theorem 10, we establish that the infinite dimensional system~\eqref{eqn:infcarl} is asymptotically stable with the learned $u$.


To examine the discrenpancy between the trajectories of~\eqref{eqn:Ncarl} and~\eqref{eqn:infcarl}, we consider two assumptions from~\cite{carlerror}.
We assume that the expansion coefficients $A_{1k}$ satisfy the exponential decay property that $|A_{1k}| \leq D_0 H^{-k}, \, k \geq 1$, for positive constants $D_0$ and $H$. We also assume that $A_{11}$ is a stable matrix. Next, we obtain a generalized version of Theorem 3.5 in~\cite{carlerror} valid for $||\psi_{\mathbf{N}}^{[k]} - x^{[k]}||, \, \forall k \geq 1$.
\begin{lemma}\label{th:carlerrorbound}
Given that $||x_0||_2 < H \epsilon_0 \lambda_{min}(S)$, $|A_{1k}| \leq D_0 H^{-k}, \, k \geq 1$, and $A_{11}$ is a Hurwitz matrix such that $A_{11}^TS + SA_{11} \leq \mu_0$,  for a positive definite matrix $S$ and $\mu_0 >0$, the error bound for each Carleman state satisfies $\psi_{\mathbf{N}}^{[1]},\hdots \,, \psi_{\mathbf{N}}^{[N]}$
\begin{equation}
    ||\psi_{\mathbf{N}}^{[k]} - x^{[k]}|| 
    \leq \dfrac{k\epsilon_0^k||x_0||_2H^{k-1}}{\omega\lambda_{min}(S)(1 - \epsilon_0)} \Bigg( \dfrac{||x_0||_2}{H\epsilon_0 \lambda_{min}(S)} \Bigg)^{\mathbf{N}}
    \label{eqn:carlerror}
\end{equation}
where $\omega = H \kappa/CD_0$, $\epsilon_0 = \omega/(1+\omega) $ and $C>0$.
\end{lemma}
\begin{proof}
Since this is not the focus of our paper, we briefly present the proof by modifying the steps in \cite{carlerror} for a general case when $A_{11}$ is Hurwitz. Replacing $||x(t)||_2$ by $x(t)^T S x(t)$ in (5.41) of \cite{carlerror}, we get the following bound 
\begin{equation}
\begin{split}
    \dfrac{d (x^T S x)}{dt} &\leq -2\Bigg(\mu_0 - \dfrac{||S||D_0||x_0||_2}{H(H-||x_0||_2)}\Bigg)||x||_2^2\\
    &\leq -\dfrac{2}{\lambda_{max}(S)}\Bigg(\mu_0 - \dfrac{||S||D_0||x_0||_2}{H(H-||x_0||_2)}\Bigg) (x^TSx).
\end{split}
\end{equation}
Therefore, we have the bound on $||x(t)||_2$ as follows
{\footnotesize
\begin{equation}
    ||x(t)||_2 \leq \dfrac{||x_0||_2}{\lambda_{min}(S)} \exp{\Bigg(-\dfrac{1}{\lambda_{max}(S)} \Big( \mu_0 - \dfrac{||S||D_0||x_0||_2}{H(H-||x_0||_2)}\Big) t \Bigg).}
\end{equation}}
\normalsize
Next, we use the fact that $A_{kk}$ is stable (since $A_{11}$ is stable), and $\text{Re}\{\lambda(A_{11})\} < -\kappa$ for $\kappa>0$ to obtain
\begin{equation}
    ||\exp(A_{kk}t)|| \leq C_k e^{-k \kappa t} \leq C e^{-k \kappa t},
\end{equation}
where $C = \max (C_1, C_2, \hdots, C_{\mathbf{N}})$, and $\text{Re} \{\lambda(A_{kk})  \} = k\text{Re}\{\lambda(A_{11})\} \leq - k\kappa$. 
Therefore, (5.54) in \cite{carlerror} is converted to \scriptsize
\begin{equation}
\begin{aligned}
    ||\eta_{k,\mathbf{N}}||_{\infty} &\leq \dfrac{CD_0k}{H} \sum_{l=k+1}^{\mathbf{N}}H^{k-l} \int_{t_0}^t e^{-k\kappa(t-s)} ||\eta_{l,\mathbf{N}}||_{\infty} ds\\
    &+ \dfrac{CD_0k \bigg( \dfrac{||x_0||_2}{\lambda_{min}(S)}\bigg)^{\mathbf{N}+1}}{H\bigg(H - \dfrac{||x_0||_2}{\lambda_{min}(S)}\bigg)} \int_{t_0}^t e^{-k \kappa (t-s)}H^{k-\mathbf{N}}ds
    \end{aligned}    
\end{equation}\normalsize
Therefore, making similar modifications to steps (5.55-5.57) in \cite{carlerror}, we get the bound in \eqref{eqn:carlerror}. 
\end{proof}
From Lemma \ref{th:carlerrorbound}, we see that when $||x_0||_2 < H\epsilon_0 \lambda_{min}(S)$,  the error goes to zero for $\mathbf{N}\rightarrow \infty$.  Thus, for a sufficiently large $\mathbf{N}$ and $||x_0||_2 < H\epsilon_0 \lambda_{min}(S)$, the vector $\psi_{\mathbf{N}}$ is close to the states of the actual system $\psi^x_{\mathbf{N}} = [x^T \, x^{[2]T}\,\hdots \, x^{[\mathbf{N}]T}]^T$ i.e., $||\psi_{\mathbf{N}} - \psi^x_{\mathbf{N}}|| \sim \mathcal{O}(\zeta)$, where $\zeta$ is a small positive quantity. 

In Section~\ref{subsec:trunstab}, we have formulated our theory based on the availability of $\psi_{\mathbf{N}}$. However in practice, we have access to $\psi^x_{\mathbf{N}}$. Since $||\psi_{\mathbf{N}} - \psi^x_{\mathbf{N}}|| \sim \mathcal{O}(\zeta)$, it follows from Theorem 2 in \cite{sayakbai} that the learned cost matrix $P^x_{\mathbf{N}} = P_{\mathbf{N}} + \mathcal{O}(\zeta)$, where $P_{\mathbf{N}}$ is the cost matrix if we were to use $\psi_{\mathbf{N}}$ for learning. Therefore, \eqref{eqn:Vibar} gets modified to
\begin{multline}
    \dot{\bar{V}}_i^{\epsilon} = \psi_{\mathbf{N}}^T {\bar{P}^{x (i)}_{\mathbf{N}}} (\mathcal{A}_{\mathbf{N}}\psi_{\mathbf{N}} - B_{\psi_{\mathbf{N}}}R^{-1}B_{\psi^x_{\mathbf{N}}}{\bar{P}^{x (i)}_{\mathbf{N}}}\psi^x_{\mathbf{N}}) \\ + (\mathcal{A}_{\mathbf{N}}\psi_{\mathbf{N}} - B_{\psi_{\mathbf{N}}}R^{-1}B_{\psi^x_{\mathbf{N}}}{\bar{P}^{x (i)}_{\mathbf{N}}}{\psi^x}_{\mathbf{N}})^T{\bar{P}^{x (i)}_{\mathbf{N}}}\psi_{\mathbf{N}}.
\end{multline}
Replacing $\psi^x_{\mathbf{N}} = \psi_{\mathbf{N}} + \mathcal{O}(\zeta)$ and ${\bar{P}^{x (i)}_{\mathbf{N}}} = \bar{P}^i_{\mathbf{N}} + \mathcal{O}(\zeta)$, and further ignoring higher-order terms $\mathcal{O}(\geq \zeta^2)$ and simplifying, we obtain
\begin{equation}
    \dot{\bar{V}}_i^{\epsilon} \leq -\psi_{\mathbf{N}}^T \Gamma \psi_{\mathbf{N}} +2\gamma \|\psi_{\mathbf{N}}\|^2 + \underbrace{\psi_{\mathbf{N}}^T\bar{P}^i_{\mathbf{N}}B_{\psi_{\mathbf{N}}}^T \mathcal{O}(\zeta)}_{<\beta \mathcal{O}(\zeta) ||\psi_{\mathbf{N}}||^2  },
    \end{equation}
    where $\beta > 0$, which means
\begin{equation}
     \dot{\bar{V}}_i^{\epsilon} \leq  -\psi_{\mathbf{N}}^T (\lambda_{min}(\Gamma) - 2\gamma - \beta \mathcal{O}(\zeta) ) \psi_{\mathbf{N}}.
\end{equation}
We then conclude that given a large $\lambda_{min}(\Gamma)$,  $\dot{\bar{V}}_i^{\zeta} < 0$ and thus~\eqref{eqn:Ncarl} is asymptotically stable. 

\begin{theorem}
    Given that the truncated closed-loop system~\eqref{eqn:Ncarl} is locally asymptotically stable, the infinite dimensional nonlinear system~\eqref{eqn:infcarl} is also locally asymptotically stable.
\end{theorem}
\begin{proof}
    Since the truncated closed-loop system is locally asymptotically stable, 
    the closed-loop matrix $A_{11,cl}$ in (53) is Hurwitz. Since the stability of the infinite dimensional system is also dictated by $A_{11,cl}$, the infinite dimensional system is also locally asymptotically stable~\cite{localstab}.
\end{proof}
Fig. \ref{fig:flowdiagram} summarizes the important steps in proving the stability of the learnt truncated controller for the infinite dimensional system. 
\begin{figure}[h]
  \centering
  \includegraphics[width=\linewidth]{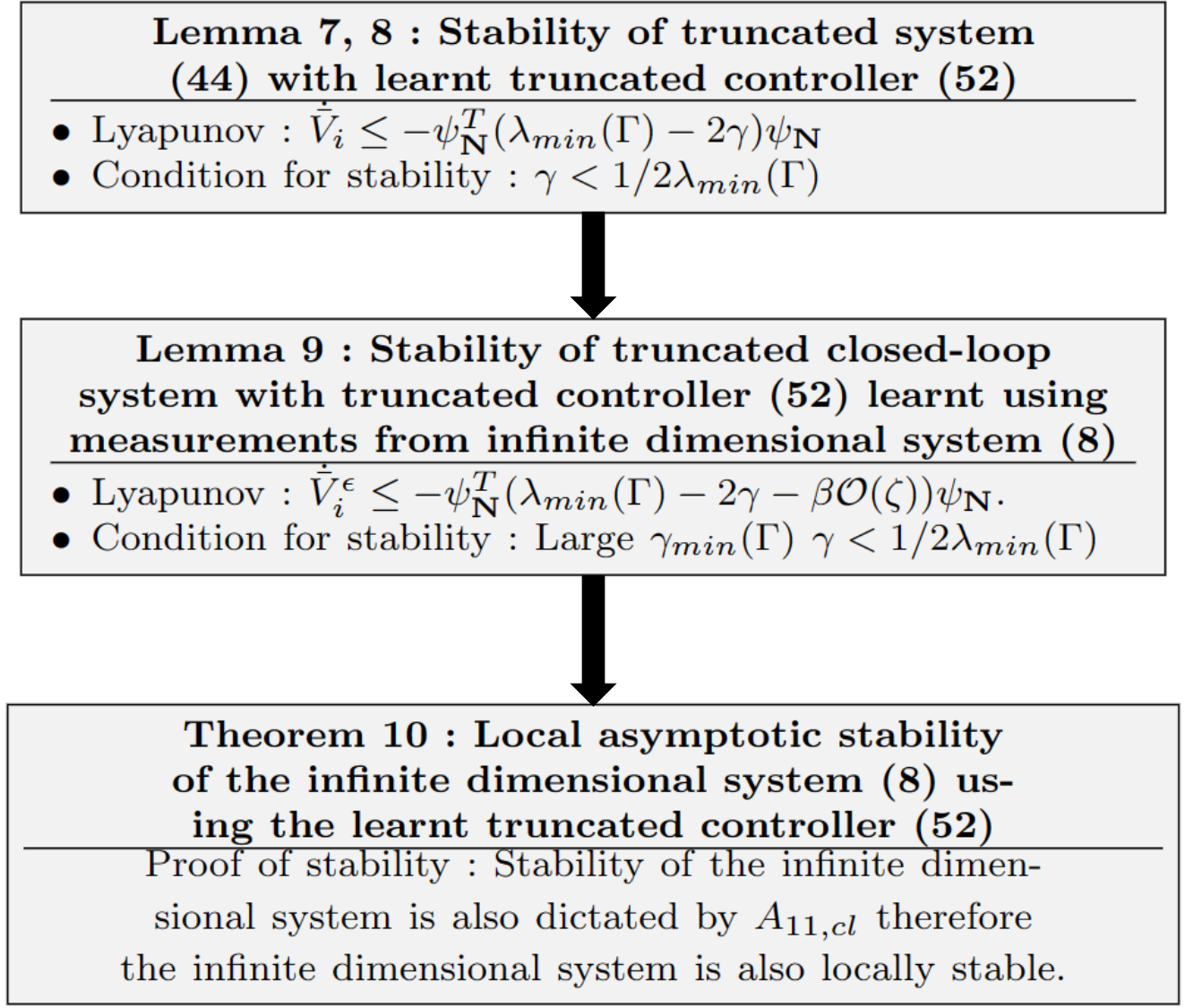}
\caption{Flowchart for stability proof of truncated controller}
   \label{fig:flowdiagram}
\end{figure}
\section{Online implementation}
In this section, we present the steps for online implementation of the proposed control design. We propose two variants for the implementation, namely, \textit{on-policy}, and \textit{off-policy}. The on-policy method involves learning the controller while simultaneously implementing it in the system at each iteration, whereas the off-policy method first learns the controller until a certain number of iterations and thereafter implements it following the final iteration. 

\subsection{On-policy implementation}

We denote $T$ as the duration after which the matrix $P_i$ is updated. Therefore, from time $[t_0 + iT , t_0+(i+1)T )]$, the control input corresponding to $P_i, i=0,1,\hdots,N$ is actuated. Denote the sampling time as $\delta t$. 
Let us denote $\Psi_{\mathbf{N}(i)}$ and $Y_{\mathbf{N}(i)}$ as the following matrices \scriptsize
\begin{multline}
    \Psi_i = \begin{bmatrix}
    \bar{\psi}_{\mathbf{N}}(t_0 + iT) - \bar{\psi}_{\mathbf{N}}(t_0 + iT + \delta t) \\
    \vdots \\
    \bar{\psi}_{\mathbf{N}}(t_0 + (i+1)T - \delta t) - \bar{\psi}_{\mathbf{N}}(t_0 + (i+1)T )
    \end{bmatrix}, \\
    Y_i \hspace{-0.1cm} = \hspace{-0.1cm} \begin{bmatrix}
     \int\limits_{t_0 + iT}^{t_0 + iT + \delta t} \Big( x(s)^T Q_1 x(s)  + \psi_{\mathbf{N}}(s)^TK_{\mathbf{N}}^TRK_{\mathbf{N}}\psi_{\mathbf{N}}(s) \Big) ds \\
     \vdots \\
     \int\limits_{t_0 + (i+1)T- \delta t}^{t_0 + (i+1)T } \Big( x(s)^T Q_1 x(s)  + \psi_{\mathbf{N}}(s)^TK_{\mathbf{N}}^TRK_{\mathbf{N}}\psi_{\mathbf{N}}(s) \Big) ds
    \end{bmatrix}.
\end{multline}\normalsize
Therefore, the following must hold
\begin{equation}
   \bar{p}_{\mathbf{N}(i)}^T \Psi_i = Y_i,
   \label{eqn:ppsi=y}
\end{equation}
where $\bar{p}_{\mathbf{N}(i)}$ is such that $\bar{p}_{\mathbf{N}(i)}\bar{\psi}_{\mathbf{N}} = \psi_{\mathbf{N}}^T P_{\mathbf{N}(i)} \psi_{\mathbf{N}}$.
We can obtain the least square solution of \eqref{eqn:ppsi=y} in terms of the Moore-Penrose pseudo inverse as $\bar{p}_{\mathbf{N}(i)}^T = Y_i \Psi_i^{\dagger}$, where $\Psi_i^{\dagger} = \Psi_i^T(\Psi_i \Psi_i^T)^{-1}$. To excite all system modes, and to prevent ill-conditioning of matrices in computing the least-squares solutions, we input a persistently exciting noise $u_{\text{noise}}(t)$. In that case, \eqref{eqn:psiPpsi} is transformed as 
\begin{multline}
    \dfrac{d \big(\psi_{\mathbf{N}}^TP_{\mathbf{N}(i)}\psi_{\mathbf{N}} \big)}{dt} = \psi_{\mathbf{N}}^T (\mathcal{A}_{cl,i-1}^TP_{\mathbf{N}(i)} +  P_{\mathbf{N}(i)}\mathcal{A}_{cl,i-1})\psi_{\mathbf{N}} \\ + 2 \psi_{\mathbf{N}}^T P_{\mathbf{N}(i)} B_{\psi_{\mathbf{N}-1}}u_{\text{noise}}.
\end{multline}
We can represent $\psi_{\mathbf{N}}^T P_{\mathbf{N}(i)} = \bar{p}_{\mathbf{N}(i)}^T \beta (\psi_{\mathbf{N}})$, where $\beta(\psi_{\mathbf{N}})$ is a matrix of appropriate dimension formed from $\psi_{\mathbf{N}}$. Therefore, $\Psi_i$ can be modified to $\Psi_{i\text{(noise)}}$ as 
\begin{equation}
\Psi_{i\text{(noise)}} = \Psi_i + \begin{bmatrix}
\int_{t_0 + iT}^{t_0 + iT + \delta t} 2\beta(\psi_{\mathbf{N}})B_{\psi_{\mathbf{N}-1}}u_{\text{noise}}\\
\vdots \\
\int_{t_0 + (i+1)T + \delta t}^{t_0 + (i+1)T + \delta t} 2\beta(\psi_{\mathbf{N}})B_{\psi_{\mathbf{N}-1}}u_{\text{noise}} \nonumber
\end{bmatrix}.
\end{equation}
Algorithm \ref{alg:complete} summarizes the steps for on-policy implementation. We drop the truncation notation $\mathbf{N}$ from $P_i$ for simplicity.

\begin{algorithm}[h]
  \begin{algorithmic}[1]
 \STATE Start with an initializing $P_0$ such that (1) is stable.
 \STATE \textbf{For} $i=0,..,L$
 \STATE \, Actuate input \footnotesize $  u = - R^{-1}(B_1^T P_{11} x + B_1^T P_{12} \psi_{\mathbf{1-N}} + \psi_{\mathbf{N-1}}^T B_2^TP_{21} x ) + u_{
 \text{noise}}, \forall t \in [t_0 + iT , t_0+(i+1)T )] $\normalsize
 \STATE \, Form $\psi_{\mathbf{N}}$ from $x$, use it to form $\Psi_{i(\text{noise})}$, $Y_i$
 \STATE \, Update $P_{i+1}$ from the solution of $\bar{p}^T_{i+1} = Y_i \Psi^{\dagger}_{i\text{noise}}$
 \STATE \, If $||P_i - P_{i+1}|| > \epsilon$
 \STATE \,\,\,\, then $i \longleftarrow i+1$ 
 \STATE \, Else
 \STATE \,\,\,\, $P^{*} = P_{i+1}$
 \STATE End
  \end{algorithmic}
  \caption{On-policy implementation}
  \label{alg:complete}
\end{algorithm}


\subsection{Off-policy implementation}
Consider the open-loop system (assumed to be stable) 
\begin{equation}
    \dot{\hat{\psi}}_{\mathbf{N}} = \mathcal{A}_{\mathbf{N}}\hat{\psi}_{\mathbf{N}} + \begin{bmatrix}
    B_1 \\ B_2 \hat{\psi}_{\mathbf{N-1}}
    \end{bmatrix} u,
    \label{eqn:offpolicy1}
\end{equation}
which can be written  as 
\begin{equation}
     \dot{\hat{\psi}}_{\mathbf{N}} = \mathcal{A}_{\mathbf{N}}\hat{\psi}_{\mathbf{N}} -B(K_{\mathbf{N}}) \hat{\psi}_{\mathbf{N}} + B(K_{\mathbf{N}}) \hat{\psi}_{\mathbf{N}} + \begin{bmatrix}
    B_1 \\ B_2 \hat{\psi}_{\mathbf{N}-1}
    \end{bmatrix} u,
\end{equation}
where $B(K_{\mathbf{N}})$ from \cite{motee2} (affine in $K_{\mathbf{N}}$, proof of Lemma 7) is given by
\begin{equation}
  B(K_{\mathbf{N}}) =     \begin{bmatrix}
 B_1 K_{1\mathbf{N}} & B_1 K_{2\mathbf{N}} & \hdots \\ 
      0 & B_2 \otimes K_{1\mathbf{N}} & \hdots \\
      \vdots & \vdots & \vdots \\
    \end{bmatrix}.
\end{equation}
As in Section VI-A, we input an exploratory noise $u(t) = u_{\text{noise}}(t)$. We write the derivative of the Lyapunov function as 
\begin{multline}
 \dfrac{d \big(\hat{\psi}_{\mathbf{N}}(t)^TP_{\mathbf{N}(i)}\hat{\psi}_{\mathbf{N}}(t) \big)}{dt} = \hat{\psi}_{\mathbf{N}}^T (\mathcal{A}_{cl,i}^TP_i + P_i \mathcal{A}_{cl,i} )  \hat{\psi}_{\mathbf{N}} \\ + 2 \hat{\psi}_{\mathbf{N}}^TP_i \Big( B(K_{\mathbf{N}}) \hat{\psi}_{\mathbf{N}} + \begin{bmatrix}
 B_1 \\ B_2 \hat{\psi}_{\mathbf{N}-1}
 \end{bmatrix} u_{\text{noise}} \Big),
\end{multline}
where $\hat{\psi}_{\mathbf{N}}$ is the trajectory for the open-loop system \eqref{eqn:offpolicy1} with $u=u_{\text{noise}}$. Therefore, $\Psi_i$ can be modified as \scriptsize
\begin{multline}
    \Psi_i^{off} = \Psi_i - \\ \begin{bmatrix}
\int_{t_0 + iT}^{t_0 + iT + \delta t} 2\beta(\psi_{\mathbf{N}})\big( B(\psi_{\mathbf{N}-1})u_{\text{noise}} + B(K_{\mathbf{N}})\psi_{\mathbf{N}}\big)\\
\vdots \\
\int_{t_0 + (i+1)T + \delta t}^{t_0 + (i+1)T + \delta t} 2\beta(\psi_{\mathbf{N}})\big( B(\psi_{\mathbf{N}-1})u_{\text{noise}} + B(K_{\mathbf{N}})\psi_{\mathbf{N}}\big) \nonumber
\end{bmatrix}.
\end{multline}\normalsize
Following is the list of changes that need to be made to Alg. \ref{alg:complete} to make it off-policy:
\begin{enumerate}
    \item In step 3, $u=u_{\text{noise}}$ in learning phase till convergence.
    \item In step 4, replace $\Psi_i$ by $\Psi^{off}_i$
    \item Once $P_i$ converges, or iterations reach end, actuate the learned control $u = - R^{-1}(B_1^T P^{*}_{11} x + B_1^T P^{*}_{12} \psi_{\mathbf{1-N}} + \psi_{\mathbf{N-1}}^T B_2^TP^{*}_{21} x )$.
\end{enumerate}

\noindent For both on-policy and off-policy methods, the truncation $\psi$ can be of any order. However, for our simulations in the next section, we will use $N=2,3$ to show its advantage over other optimal control methods for linear and nonlinear systems.

\section{Structured Carleman RL Control}
{\color{black}
Designing a fully dense Carleman controller may not always be feasible for implementation, due to bandwidth and communication topology limitations. For a multi-agent system, this means requiring an all-to-all communication topology which might be expensive to build and maintain. To tackle this practical challenge, we present a structured version of our proposed Carleman based RL controller, which can be designed in a data-driven manner according to a desired topology based on the actual communication topology, provided that a stabilizing controller for the topology exists.
}
\subsection{Recapitulation of model-based design}
Note that, according to \cite{ricattilike}, the Lyapunov equation for the Carleman linearized system (4) can be written in a Riccati-like equation as 
\begin{equation}
    \mathcal{A}^T P + P\mathcal{A} + Q - (P \begin{bmatrix}
    B_1 \\ 0
    \end{bmatrix} + W_P)R^{-1}(\begin{bmatrix}
    B_1 \\ 0
    \end{bmatrix}^TP + W_P^T) = 0.
    \label{eqn:carlricatti}
\end{equation}
    \begin{equation}
        u^{*} = -R^{-1} ( W_P + \begin{bmatrix}
        B_1 \\0
        \end{bmatrix}^T P) \psi. 
    \end{equation}

\begin{theorem}
For an arbitrary matrix $L$, we can find a positive definite $P$ for the equation
\begin{multline}
     \mathcal{A}^T P + P\mathcal{A} + Q - (P \begin{bmatrix}
    B_1 \\ 0
    \end{bmatrix} + W_P)R^{-1}(\begin{bmatrix}
    B_1 \\ 0
    \end{bmatrix}^TP + W_P^T) \\ + L^TRL = 0,
    \label{eqn:genricatti}
\end{multline}
and the control feedback gain $K = \Pi_P - L$, where 
 $\Pi_P$ is such that 
\begin{equation}
\Pi_P \psi = R^{-1} ( \psi^T W_P + \begin{bmatrix}
        B_1 \\0
        \end{bmatrix}^T P \psi).
\end{equation}

\end{theorem}
\begin{proof}
Consider the Lyapunov function $V_L = x^T P x$. Then we have
\begin{multline}
    \dot{V}_L = \psi^T P \Big( \mathcal{A}\psi + B_{\psi} \big(-R^{-1}(\psi^TW_P + \begin{bmatrix}
    B_1 \\ 0
    \end{bmatrix}^T P \psi) + L \psi \big) \Big) \\
    + \Big( \mathcal{A}\psi + B_{\psi} \big(-R^{-1}(\psi^TW_P + \begin{bmatrix}
    B_1 \\ 0
    \end{bmatrix}^T P \psi) + L \psi \big) \Big)^T P \psi \\
    = \psi^T \Big( \mathcal{A}^TP + P\mathcal{A} - B_{\psi} \Pi^T_P + PB(\psi)L - \Pi_P B_{\psi}^T \\ + L^TB_{\psi}^TP    \Big) \psi \\
    = \psi^T \Big( -Q + \Pi_P^TR\Pi_P - L^TRL - PB_{\psi}\Pi_P^T - \Pi_PB_{\psi}^TP  \Big) \psi  \\
\end{multline}
Recall that $\psi^TPB(\psi)\psi = \psi^T \Pi_P \psi $. Therefore, we have
\begin{multline}
    \dot{V}_L = \psi^T \Big( -Q - (\Pi_P^T - L^T)R(\Pi_P - L)  \Big) \psi 
    \\=  -x^T Q_1 x - \psi^T (\Pi_P^T - L^T)R(\Pi_P - L) \psi.
\end{multline}
\end{proof}
Note that since $Q_1>0$, and $R>0$, the RHS is negative definite. Therefore $V_L < 0$.

Next, based on \cite{geromel} and \cite{sparsesayak}, the following iterative steps are presented to compute $K$ when the model is known. Let $\Omega$ be the binary matrix denoting the desired structure of $K$. Therefore, if $K_{i,j}$ is desired to be 0, $\Omega_{i,j} = 0$, else $\Omega_{i,j}=1$. Denote $F(\Pi_P) = \Pi_P \odot \Omega^c$, where $\Omega^c = 1 - \Omega$ denotes the complement of $\Omega$.
\subsubsection*{Steps}
\begin{enumerate}
    \item Set the iteration index $i=0$ and $L=0$.
    \item Solve the Riccati equation \eqref{eqn:genricatti}
    \item Compute $L_{i+1} = F(\Pi_P)$. If $||L_{i+1} - L_i|| < \epsilon$, go to step 4 else set $i=i+1$ and return to step 2.
    \item Determine feedback gain $K$ as 
    \begin{equation}
        K = \Pi_{P_{i+1}} - L_{i+1}.
    \end{equation}
\end{enumerate}

\subsection{Model-free design}

Next, we present a lemma which would be important to present the model-free policy iteration based method.
\begin{lemma}
Solving the Riccati like equation \eqref{eqn:genricatti} is equivalent to solving the Lyapunov equation \eqref{eqn:carllyapunov} with $K = \Pi_P - L$.
\end{lemma}
\begin{proof}
Consider pre and post multiplying \eqref{eqn:genricatti} by $\psi^T$ and $\psi$. We have
\begin{multline}
    \psi^T \Big( 
     \mathcal{A}^T P + P\mathcal{A} + Q - (P \begin{bmatrix}
    B_1 \\ 0
    \end{bmatrix} + W_P)R^{-1}(\begin{bmatrix}
    B_1 \\ 0
    \end{bmatrix}^TP\\  + W_P^T) + L^TRL  \Big) \psi = 0.
\end{multline}

Note, that $L = \Pi_P - K$. Therefore, we have
\begin{multline}
    \psi^T \Big( 
     \mathcal{A}^T P + P\mathcal{A} + Q - (P \begin{bmatrix}
    B_1 \\ 0
    \end{bmatrix} + W_P)R^{-1}(\begin{bmatrix}
    B_1 \\ 0
    \end{bmatrix}^TP + W_P^T) \\ + (\Pi_P - K)^TR(\Pi_P-K)  \Big) \psi = 0.  
\end{multline}

Recalling $\Pi_P$ from (38), we can write
    \begin{equation}
      \psi^T \Big(   \mathcal{A}^TP + P\mathcal{A} + Q - \Pi_P^TRK - K^TR\Pi_P + K^TRK\Big) \psi = 0.
    \end{equation}
Recall that $\psi^T P\begin{bmatrix}
    0 \\ B_2
    \end{bmatrix} \psi = W_P \psi$, and $\Pi_P \psi$ from (38), we get the Lyapunov equation

\begin{equation}
    \mathcal{A}_{cl}^TP + P\mathcal{A}_{cl} = -(Q + K^TRK), \,\, \forall \psi
\end{equation}
\end{proof}

Now, we are ready to present the model-free version of the structured iterative algorithm.
\subsubsection*{Steps}
\begin{enumerate}
    \item Set iteration index $i=0$, $L=0$ and a stabilizing $K_0$.
    \item For index $i$ and $K_i$, solve for $P_i$ in (13).
    \item Update the controller gain $K_{i+1} = \Pi_{P_i} - L_i$.
    \item Update $L_{i+1} = F(\Pi_{P_i})$. If $||L_{i+1} - L_i|| < \epsilon$, stop iterations else return to step 2.
\end{enumerate}

Note that since we have already established that solving the Riccati like equation \eqref{eqn:genricatti} is equivalent to solving \eqref{eqn:carllyapunov} using Lemma 5, step 2 of the model based method is equivalent to step 2 of the above method. Since all functions are smooth, the convergence follows from \cite{ricattilike} and \cite{sparsesayak}. Results are presented in Section VIII(B).

\subsection{$N^{th}$ order truncation}
Note that, the discussion so far has been for the design of structured controllers for the infinite dimensional Carleman system. However, for practical implementation, we need to have a finite truncated version of the above steps. 

In Section 5, we already showed that if the condition \eqref{eqn:truncerror} is satisfied, then even for the truncated system, if the updates are made according to Section IV(B), then the controllers at all iterations are stabilizing and solves the truncated Lyapunov equation. Therefore, step 2 always yields a stabilizing controller given than \eqref{eqn:truncerror} is met. 

Next, we show that a truncated version of the Riccati like equation \eqref{eqn:genricatti} holds for certain condition, which would imply that step 3 also holds and leads to convergence of the controller to the desired structure. Consider the truncated Riccati like equation
\begin{multline}
       \mathcal{A}_N^T P_N + P_N\mathcal{A}_N + Q_N - (P_N \begin{bmatrix}
    B_1 \\ 0
    \end{bmatrix} + W_{P_N})R^{-1}(\begin{bmatrix}
    B_1 \\ 0
    \end{bmatrix}^TP_N \\ + W_{P_N}^T) + L_N^TRL_N  = 0
    \label{eqn:trunccarlricatti}
\end{multline}

Similar to theorem 4, consider a Lyapunov function $V_{L_N} = \psi_N^T P_N \psi_N$, then the Lyapunov derivative is
\begin{multline}
    \dot{V}_{L_N} =  \psi_N^T P_N \Big( \mathcal{A}_N\psi_N + B_{\psi_{\mathbf{N}}} \big(-R^{-1}(\psi_N^TW_{P_N} \\  + \begin{bmatrix}
    B_1 \\ 0
    \end{bmatrix}^T P_N \psi_N) + L_N \psi_N \big) \Big) 
    + \Big( \mathcal{A}_N\psi_N +   \\ \big(-R^{-1}B_{\psi_{\mathbf{N}}}(\psi_N^TW_{P_N} + \begin{bmatrix}
    B_1 \\ 0
    \end{bmatrix}^T P_N \psi_N) + L_N \psi_N \big) \Big)^T P_N \psi_N \\
\end{multline}
On doing some algebraic manipulations, we can write the derivative as 
\begin{multline}
    \dot{V}_{L_N} = \psi_N^T \Big( -Q_N - (\Pi_{P_N}^T - L_N^T)R(\Pi_{P_N} - L_N)  \Big) \psi_N \\
    + 2 \underbrace{\psi_N^T \big( \Pi_{P_N}^TR - P_{\mathbf{N}}B_{\psi_{\mathbf{N}}} \big)}(L_N - \Pi_{P_N}) \psi_N.
\end{multline}
In the infinite dimensional case, $R^{-1}B_{\psi}P\psi = \Pi_{P}\psi$, therefore the underbraced term becomes 0. Moreover, it is easy to see that the underbraced term for the truncated case is $\psi_N^T \big( \Pi_{P_N}^TR - P_{\mathbf{N}}B_{\psi_{\mathbf{N}}} \big) = -\epsilon_{trunc}^T(P_{\mathbf{N}})R$, where $\epsilon_{trunc}$ is from \eqref{eqn:truncexp}. Note that similar to the proof of Lemma \ref{lem:lemma8}, since the sign indefinite term here is $\Delta^s = 2 \psi_N^T \big(\epsilon_{trunc}^T(P_{\mathbf{N}})R \big)(L_N - \Pi_{P_N}) \psi_N$, $\mathcal{O}(\Delta^s) > \mathcal{O}(\psi_{\mathbf{N}}^2)$. Hence, following similar proof from Lemma \ref{lem:lemma8}, we can show that there exists $d_2>0$ such that $\forall ||\psi_{\mathbf{N}}|| < d_2$, the truncated Riccati-like equation \eqref{eqn:trunccarlricatti}, gives a stabilizing controller $K_{\mathbf{N}} = \Pi_{P_{\mathbf{N}}} - L$.

\section{Sparse Carleman RL Control}
As shown in Section VII, the method of designing a structured control for a Carleman system designs a controller $K$ in the Carleman space which is consistent with a desired topology $\Omega$. However, the method cannot be used to design a sparse controller with a cardinality constraint, where the goal is to minimize the number of non-zero entries in the control gain $K$ for the Carleman system to derive a simple stabilizing control. {\color{black} Such objectives are important in applications having bandwidth constraints, where even though communication links might be present, the bandwidth of several links might be low, leading to delays in data transmission. The proposed sparse algorithm tries to design a simpler Carleman controller (as compared to the fully dense controller) without significantly compromising on the closed-loop performance. As we will see in this section, and validate later in the results, that the proposed method can save bandwidth with a minor drop in closed-loop performance.} In this section, we present a ADMM based method, which works in tandem with our proposed Carleman RL (section VI) to iteratively design a sparse control.

\subsection{Objective}
Let us consider the $i^{th}$ iteration and the corresponding learnt cost matrix as $P_i$. We want to learn a controller $K_{i+1}$ which minimizes the cardinality of $K_{i+1}$. We recall the generalized Riccati like equation \eqref{eqn:carlricatti} in the Carleman space. We proved that iteratively setting $K= \Pi(P) - L$ to solve \eqref{eqn:carlricatti} either in a model-based or data-driven way yields the desired control. As shown in \cite{abhishek}, the loss in optimality can be given by
\begin{equation}
    || J - J_L|| < \beta ||L||_F^2,
\end{equation}
where $||.||_F$ denotes the Frobenius norm and $\beta$ is a constant scalar. Therefore, we define our objective function at the $i^{th}$ iteration as 
\begin{equation}
    \min : \dfrac{1}{2} ||L_{i+1}||_F^2 + \text{card}(K_{i+1}),
\end{equation}
where $\text{card}(.)$ denotes cardinality of $(.)$. However, as shown in \cite{l1norm}, such $\ell_0$ problems are difficult to solve. Therefore, we consider an equivalent $\ell_1$ formulation of the above problem as
\begin{multline}
    \min : \dfrac{1}{2}||L_{i+1}||_F^2  + \gamma \sum\limits_{m,n} W_{mn} ||K_{i+1(mn)}|| 
    \,\,,\,\,\\\text{s.t.}\,:\, K_{i+1} = R^{-1} \Bigg( \begin{bmatrix}
    B_1 \\0
    \end{bmatrix}^TP_i + W^T_{P_i} \Bigg) - L_{i+1},
    \label{eqn:l1min}
\end{multline}
where $\gamma$ is a scalar weight, and $W_{mn}$ are weights for the weighted $\ell_1$-norm. 

\subsection{ADMM-based policy iteration}
To numerically solve \eqref{eqn:l1min}, we formulate the problem as an alternating direction method of multipliers (ADMM) problem. The augmented Lagrangian for the above lasso \cite{l1norm} like problem is given by
\begin{multline}
   \textbf{minimize : } \\ h(K_{i+1},L_{i+1},\Lambda,\rho) =  \dfrac{1}{2}||L_{i+1}||_F^2  + \gamma \sum\limits_{m,n} W_{mn} ||K_{i+1(mn)}|| \\ +  \text{Tr}\big(\Lambda^T (K_{i+1} + L_{i+1} - \Pi(P_i))\big)  + \dfrac{\rho}{2} || K_{i+1} + L_{i+1} - \Pi(P_i) ||^2,
    \label{eqn:auglagrangian}
\end{multline}
where $\Lambda$ is the Lagrange multiplier matrix, and $\rho$ is the dual update step length. To solve \eqref{eqn:auglagrangian}, we compute the following partial derivatives 
\begin{multline}
    \dfrac{\partial h}{\partial L_{i+1}} = L_{i+1} + \Lambda + \rho\big(K_{i+1} + L_{i+1} - \Pi(P_i)\big),\\
    \dfrac{\partial h}{\partial K_{i+1(mn)}} = \gamma \dfrac{\partial W_{mn} ||K_{i+1(mn)}||}{\partial K_{i+1(mn)}} + \Lambda  + \rho\big(K_{i+1} + L_{i+1} \\ - \Pi(P_i)\big).
    \label{eqn:partials}
\end{multline}
We denote $z$ as the internal iteration number of the ADMM at each $i^{th}$ policy iteration. \eqref{eqn:partials} can be solved in a decoupled manner as 
\begin{multline}
    L \min \text{step} : L_{i+1}^z = \dfrac{- \big( \Lambda^{z-1} + \rho (K^{z-1} - \Pi(P_i)) \big)}{1+\rho}\\
    K \min \text{step} : K_{i+1(mn)}^z = \mathcal{S}_{\gamma W_{mn}/\rho} \big(  \Pi(P_i) - L_{i+1}^z + \dfrac{\Lambda^{z-1}}{\rho}  \big)\\
    \Lambda \text{ update step} : \Lambda^{z} = \Lambda^{z-1} + \rho \big( K_{i+1}^z + L_{i+1}^z - \Pi(P_i) \big),
    \label{eqn:updatesparseeqn}
\end{multline}
where $\mathcal{S}_{(.)}$ is the soft thresholding function about the threshold point $(.)$. The soft thresholding function is given by
\begin{equation}
    \mathcal{S}_{\mu}(\omega) = 
\begin{cases}
\omega + \mu, \text{ for } \omega < \mu\\
0, \text{ for } -\mu\leq\omega\leq\mu\\
\omega - \mu , \text{ for } \omega > \mu
\end{cases} 
\end{equation}

The steps of policy iteration in tandem with the proposed ADMM based sparsity in summarized in Algorithm \ref{alg:sparse}.
\begin{algorithm}[h]
  \begin{algorithmic}[1]
 \STATE Start with an initializing $P_0$ and corresponding $K_1,L_1$ such that (1) is stable.
 \STATE \textbf{While} $||P_i - P_{i-1}|| > \epsilon$
 \STATE \,\,Learn $P_i$ using on or off-policy methods in Sec VI.
 \STATE \,\,Compute $\Pi(P_i) = R^{-1}(\begin{bmatrix}
 B_1 \\ 0
 \end{bmatrix}^TP_i + W_{P_i}^T) $
 \STATE \,\,Initialize $\Lambda^0, \rho$ and weights $W_{mn}$
 \STATE \,\,\textbf{While} 
\STATE \,\,\,\, Compute $K_{i+1}^{z+1}$, $L_{i+1}^{z+1}$ and $\Lambda^{z+1}$ from \scriptsize
\begin{multline*}
    L-\min \text{step} : L_{i+1}^z = \dfrac{- \big( \Lambda^{z-1} + \rho (K^{z-1} - \Pi(P_i)) \big)}{1+\rho}\\
    K-\min \text{step} : K_{i+1(mn)}^z = \mathcal{S}_{\gamma W_{mn}/\rho} \big(  \Pi(P_i) - L_{i+1}^z + \dfrac{\Lambda^{z-1}}{\rho}  \big)\\
    \Lambda-\text{ update step} : \Lambda^{z} = \Lambda^{z-1} + \rho \big( K_{i+1}^z + L_{i+1}^z - \Pi(P_i) \big),
    \label{eqn:updatesparseeqn}
\end{multline*}\normalsize
\STATE \,\,\,\, Increase step size $\rho \leftarrow \alpha \rho $, where $\alpha$ is step-size increment factor
\STATE \,\,\,\, \textbf{If} $||K_{i+1}^z - K_{i+1}^{z-1}|| > \epsilon_1$ \footnotesize \& \normalsize $||L_{i+1}^z - L_{i+1}^{z-1}|| > \epsilon_2$ 
\STATE \,\,\,\,\,\,\, Update internal iteration $z = z+1$
\STATE \,\,\,\, \textbf{Else}
\STATE \,\,\,\,\,\,\, \textit{Break while}
\STATE \,\, \textbf{End while}
\STATE \,\, Update $i = i+1$, $K_{i+1} = K_{i+1}^z , L_{i+1} = L_{i+1}^z$ and go to step 3.
\STATE \textbf{End while}
  \end{algorithmic}
  \caption{Sparse controller learning}
  \label{alg:sparse}
\end{algorithm}
Note that in Algorithm \ref{alg:sparse}, $\alpha > 1$ is the step size increment factor in step 8 which ensures that the ADMM algorithm is not trapped at a local minima by increasing the learning rate factor.
Next, we state the following theorem for convergence of $P_i$.
\begin{theorem}
Given that at each $i^{th}$ policy iteration, the ADMM steps converge for a large $z$ minimizing \eqref{eqn:auglagrangian}, then the learnt $P_i$ will converge for $i \rightarrow \infty$.
\end{theorem}
\begin{proof}
Since the above ADMM problem \eqref{eqn:auglagrangian} is strongly convex, this implies that $L_{i+1}$ is a convex function of $\Pi(P_i)$ and is well behaved. Moreover, the convergence of ADMM follows from the convergence of the classic LASSO \cite{lasso1,lasso2} with linear constraints, where the step-size is chosen accurately for faster linear convergence. Therefore, similar to Section VII, the proof follows from \cite{geromel}.
\end{proof}

For the truncated version of the sparse RL controller, since we use the same Ricatti like equation as in \eqref{eqn:trunccarlricatti}, the condition for stability will remain same as discussed in the structured case. As we will shortly see in the simulation section, demanding more sparsity with a larger $\gamma$ can lead to convergence problem due to the condition not being satisfied at each iteration.

\section{Results}\label{sec:results}

\subsection{Case 1 : Second-order Oscillator}
{\color{black}First, we consider the Van der Pol oscillator system which is a popular oscillator example with nonlinear damping, and has a history of use in noth physical and biological sciences. Here, we consider a specific case of the oscillator inspired from  \cite{motee} }
\begin{align}
    \dot{x}_1 &= x_2\\
    \dot{x}_2 &= -x_1 + \dfrac{1}{2}x_2 x_1^2 + x_1 x_2 + g(x_1, x_2) u,
    \label{eqn:case1}
\end{align}
where $g(x_1, x_2) \in \mathbb{R}$. The objective function parameters are $Q_1 = [0 \; 0 ; 0 \; 1]$  and $R=1$. For this system, the control input by solving the HJB solution is given by
\begin{equation}
    u^{*} = -g(x_1,x_2) x_2.
\end{equation}

We choose $g(x_1, x_2) = (1+x_1)$, which yields the optimal linear control input as $u^{*} = -(1+x_1)x_2$. 

We simulate the system with an exploration noise for $0\leq t \leq 10$ seconds. The sampling time is chosen as $\delta t = 0.1$ seconds, and the controller is updated every 2 seconds till $t=8.5$ seconds.
\begin{figure}[h]
    \centering
    \includegraphics[width=\linewidth]{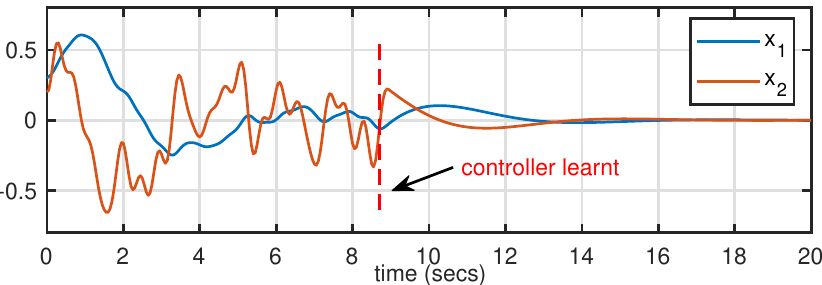}
    \caption{Learning of the $2^{nd}$ order RL controller}
    \label{fig:case1learning}
\end{figure}
\noindent Fig. \ref{fig:case1learning} shows the closed-loop trajectory for both the states $x_1$ and $x_2$ for the second-order Carleman RL control. We see that the learnt controller successfully stabilizes the system after $t=8.5$ seconds. 

Fig. \ref{fig:three graphs} shows the comparison of the learnt controller for the linear $\mathbf{N}=1$ and second order $\mathbf{N}=2$ truncation with the HJB solution. Fig. \ref{fig:oscsmalldev} shows the trajectories $x_1,x_2$ for a small initial deviation of $(0.08,0.06)$. In this case, we see that both the linear and $\mathbf{N}=2$ cases closely match the HJB solution. Fig. \ref{fig:oscbigdev} shows the results for a large initial deviation of $(0.8,0.7)$. In this case, the $\mathbf{N}=2$ controller performs significantly better than the linear control, and closely follows the HJB trajectory.
 \begin{figure}[h]
     \centering
     \begin{subfigure}[b]{0.49\linewidth}
         \centering
    \includegraphics[width=\linewidth]{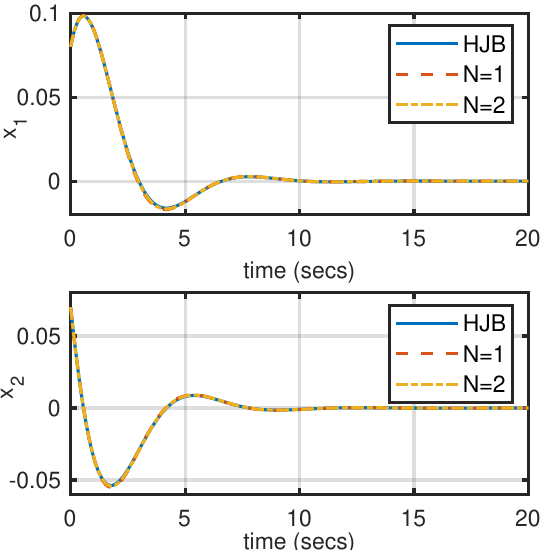}\vspace{-0.12cm}
    \caption{$x_1(0),x_2(0) = 0.08, 0.06$}
    \label{fig:oscsmalldev}
     \end{subfigure}
    \hfill
     \begin{subfigure}[b]{0.49\linewidth}
         \centering
         \includegraphics[width=\linewidth]{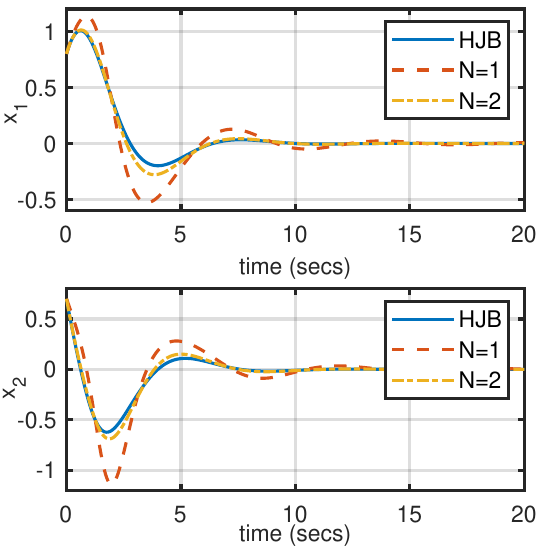}\vspace{-0.12cm}
    \caption{$x_1(0),x_2(0) = 0.8,0.7$}
    \label{fig:oscbigdev}
     \end{subfigure}
        \caption{Comparison of learnt controller with HJB solution}
        \label{fig:three graphs}
\end{figure}
Table \ref{tab:objvalue} shows the objective value for the different cases described in Fig. \ref{fig:three graphs}. We can see that for the smaller deviation case, the objective values are similar for all the three controllers, whereas for large deviations, the value of the $\mathbf{N}=2$ controller is much closer to the HJB value.
\begin{table}[h]
    \centering
    \begin{tabular}{|c|c||c|}
    \hline 
   & \multicolumn{1}{|c||}{Small deviation} & \multicolumn{1}{|c|}{Large deviation}\\
   \hline
        $\mathbf{N} = $ &  $J$ value &  $J$ value \\
        \hline 
       HJB  & 0.0119 & 2.2454 \\ \hline
       1  & 0.0119 & 2.9196  \\ \hline 
       2  & 0.0119 & 2.2720\\ \hline 
    \end{tabular}
    \caption{Objective value comparison}\vspace{-0.2cm}
    \label{tab:objvalue}
\end{table}

\subsection{Case 2 : Coordinated control of tugboat dynamics}
Next, consider an example of coordinated control for tugboats with nonlinear dynamics. Consider the dynamics of tugboat $j$ given by
\begin{equation}
    \dot{\eta}_j = R_j (\theta_j) v_j \,\, , \,\,    M_j \dot{v}_j + D_j v_j = \tau_j
\end{equation}
where $\eta_j = [x_j , y_j , \theta_j]^T$, $(x_j,y_j)$ is the position vector in an appropriate frame, and $\theta_j$ is the heading angle (yaw). $v_j = [v_{j,1} , v_{j,2} , v_{j,3}]^T \in \mathbb{R}^3$ is the velocity vector.  $R_j (\theta_j)$ is given by
\begin{equation}
    R_j = \begin{bmatrix}
    \cos{\theta_j} & -\sin{\theta_j} & 0\\
    \sin{\theta_j} &  \cos{\theta_j} & 0\\
    0 & 0 & 1
    \end{bmatrix}.
\end{equation}
The model matrices $M_j$ and $D_j$ for all tugboats are same and given by
\begin{equation}
    M_j =  \begin{bmatrix}
    33.8 & 1.0948 & 0 \\
    1.0948 & 2.764 & 0\\
    0 & 0 & 23.8
    \end{bmatrix} \, , \, D_j = \begin{bmatrix}
    7 & 0.1 & 0\\
    0.1 & 0.5 & 0\\
    0 & 0 & 2
    \end{bmatrix}.
\end{equation}
We consider a multi-agent network example with four tugboats that are connected through an all-to-all communication network to execute the state=feedback. The objective is set such that starting from random initial locations and yaw angles, the four tugboats make a square formation at the point $(x^{*}_1,y^{*}_1) = (10,10)$, $(x^{*}_2,y^{*}_2) = (10,-10)$, $(x^{*}_3,y^{*}_3) = (-10,-10)$, and $(x^{*}_4,y^{*}_4) = (-10,10)$. The goal is that at each timepoint during the travel, the boats should try to be at the same distance from the targets. Also, the velocities should asymptotically go to zero as the boats reach their destination. Define $\delta \eta_j = [x_j - x_j^{*} \, , \, y_j - y_j^{*} \,,\,\theta_j]^T $ Mathematically, the objective function is
\begin{multline}
    J = \int_0^{\infty} \sum_{j=1}^4( 5\delta\eta_j^T \delta\eta_j + v_j^Tv_j ) + \\ \underbrace{\sum_{k=1}^4 \big( \sum_{j=k+1}^4 (\delta x_j - \delta x_k)^2 + (\delta y_j - \delta y_k)^2\big)}_{\text{distance objective}}.
    \label{eqn:objfunction}
\end{multline}
The input cost matrix $R$ was chosen to be identity. We consider the Carleman states only for the interactions between yaw angle $\theta_j$ and velocities $v_{j,1} , v_{j,2}$. Therefore, for learning purposes, the matrix $R_j(\theta_j)$ for different values of $\mathbf{N}$ would be 
\begin{equation}
    \underbrace{\begin{bmatrix}
    1 & 0 & 0\\
    0 & 1 & 0\\
    0 & 0 & 1\\
    \end{bmatrix}}_{\mathbf{N}=1}  ,    \underbrace{\begin{bmatrix}
    1 & -\theta_j & 0\\
    \theta_j & 1 & 0\\
    0 & 0 & 1\\
    \end{bmatrix}}_{\mathbf{N}=2} ,     \underbrace{\begin{bmatrix}
    1 - (\theta_j^2/2) & -\theta_j & 0\\
    \theta_j & 1 - (\theta_j^2/2) & 0\\
    0 & 0 & 1\\
    \end{bmatrix}}_{\mathbf{N}=3}.
\end{equation}
Note that, this means that the higher-order Carleman states for $\mathbf{N}=2$ and $\mathbf{N}=3$ would be
$[\theta_jv_{j,1},\theta_jv_{j,2}]^T$ and $[\theta_jv_{j,1},\theta_jv_{j,2},\theta_j^2v_{j,1},\theta_j^2v_{j,2}]^T$ respectively.
\subsubsection{Learning and performance study}
\begin{figure}[h]
     \centering
     \begin{subfigure}[b]{0.49\linewidth}
         \centering
         \includegraphics[width=\textwidth,height=3cm]{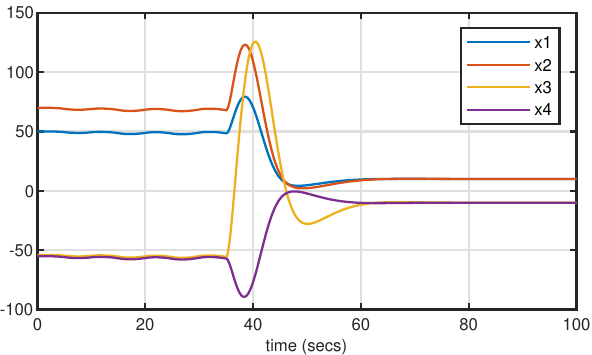}
         \caption{Position $x$}
         \label{subfig:xt}
     \end{subfigure}
     \hfill
     \begin{subfigure}[b]{0.49\linewidth}
         \centering
         \includegraphics[width=\textwidth,height=3cm]{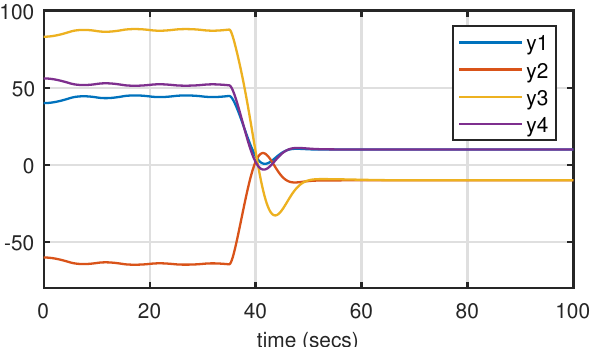}
         \caption{Position $y$}
         \label{subfig:yt}
     \end{subfigure}\hfill
     \begin{subfigure}[b]{0.49\linewidth}
         \centering
         \includegraphics[width=\textwidth,height=3cm]{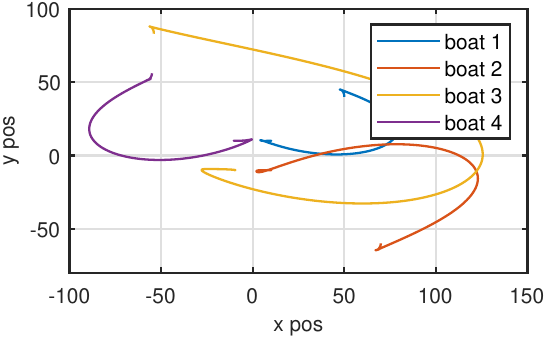}
         \caption{Trajectory $y$ vs $x$}
         \label{subfig:xvsy}
     \end{subfigure}
     \hfill
     \begin{subfigure}[b]{0.49\linewidth}
         \centering
         \includegraphics[width=\textwidth,height=3cm]{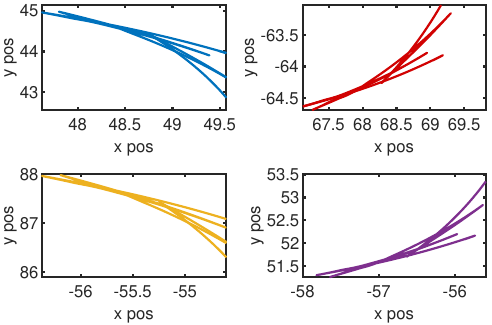}
         \caption{Learning phase}
         \label{subfig:zoomlearn}
     \end{subfigure}
        \caption{Learning and control for  $\mathbf{N}=2$ Carleman controller}
        \label{fig:N2learning}
\end{figure}

Fig. \ref{fig:N2learning} shows the learning and closed-loop performance for Carleman $\mathbf{N}=2$ case. Fig. \ref{subfig:xt}, \ref{subfig:yt} shows the $x$ and $y$ positions for the 4 boats. Fig. \ref{subfig:xvsy} shows the trajectory in the $x-y$ coordinates starting from their initial location to the final square formation. The first 38 seconds are the learning phase, where the boats are input with an alternating pulse of $+1,-1$ every 3 seconds. Fig. \ref{subfig:zoomlearn} shows the trajectory of the 4 boats in the learning phase. The learning is done such that the boats remain close to their initial conditions, and their loci towards the final location begin after the learning phase is completed.

Fig. \ref{fig:boatconvergence} shows the convergence of $P_i$ for different values of cost matrix $Q$. The first case $Q$ is corresponding to the objective function \eqref{eqn:objfunction}, and the rest are as indicated in Fig. \ref{fig:boatconvergence}. We see that for the first 2 cases, the $P_i$ converges within 3 and 5 iteration respectively. However, for the last case, where a very small value of $Q$ is used, the convergence fails. This supports the proof of Lemma 7, where we showed that a large enough $Q$ is needed to guarantee stability, and, therefore, convergence.
\begin{figure}[h]
    \centering
    \includegraphics[width=\linewidth,height=3.6cm]{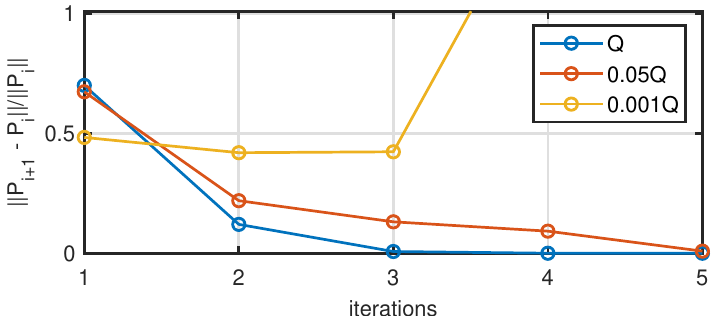}
    \caption{Convergence for varying $Q$}
    \label{fig:boatconvergence}
\end{figure}

\begin{figure*}[h]
     \centering
     \begin{subfigure}[b]{0.49\textwidth}
         \centering
         \includegraphics[width=1\textwidth]{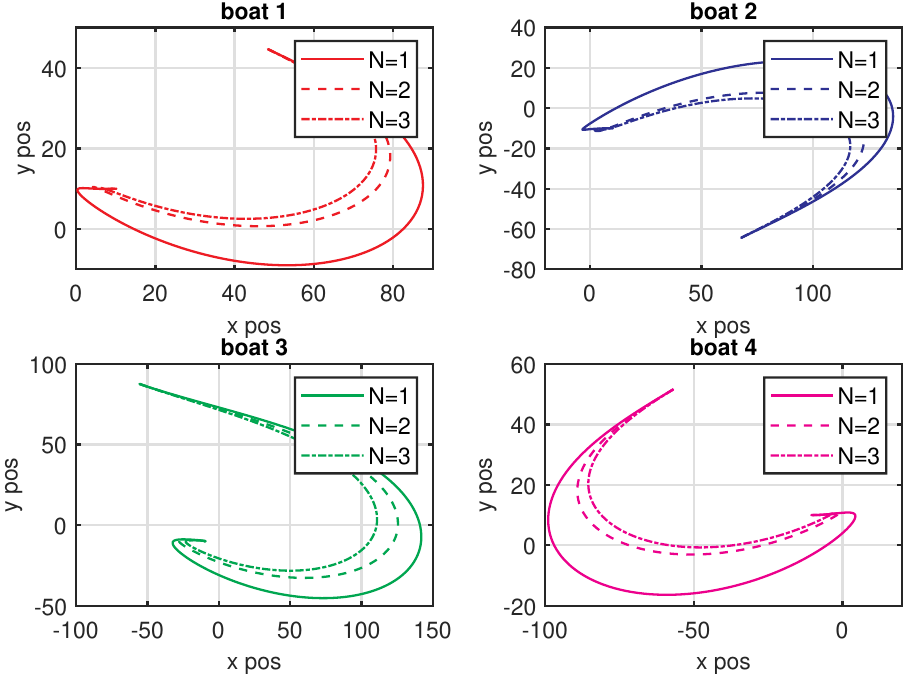}
         \caption{Small deviation in initial yaw angle $\max (\theta_j) = 0.2,\, j = 1,..,4$}
         \label{subfig:smallyaw}
     \end{subfigure}
     \hfill
     \begin{subfigure}[b]{0.49\textwidth}
         \centering
         \includegraphics[width=1\textwidth,height=6.6cm]{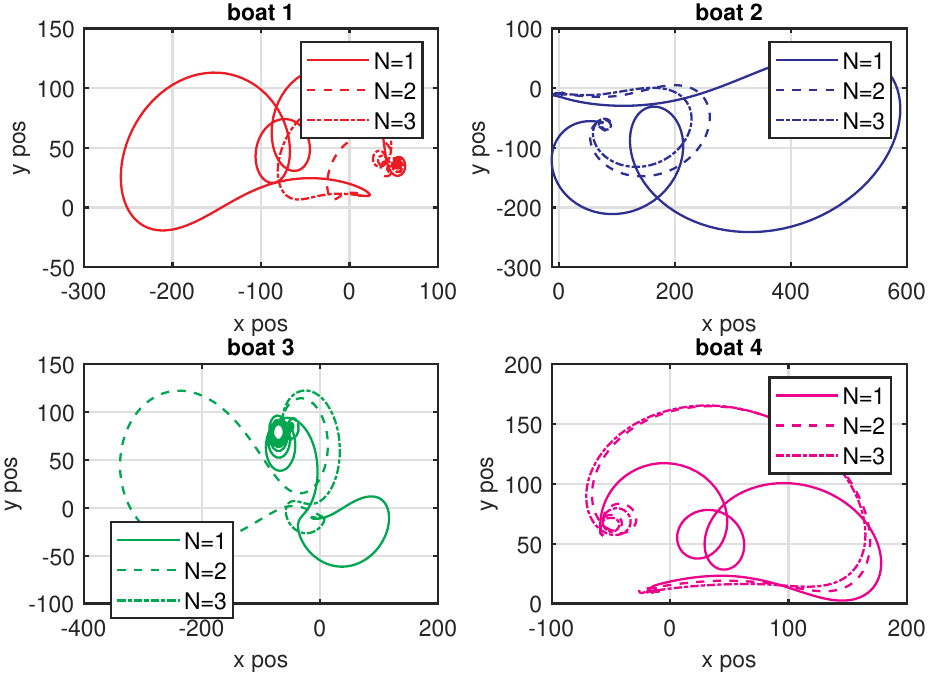}
         \caption{Large deviation in initial yaw angle $\max (\theta_j) = 3,\, j = 1,..,4$}
         \label{subfig:largeyaw}
     \end{subfigure}
        \caption{Comparison of learning based Carleman controllers for $\mathbf{N}=1, 2, 3$ truncations}
        \label{fig:Ncomparison}
\end{figure*}

Fig. \ref{fig:Ncomparison} shows the comparison of the learnt Carleman controller for truncation lengths of $\mathbf{N}=1, 2, 3$ for two different cases. Fig. \ref{subfig:smallyaw} shows the case when the initial yaw angles are small. Fig. \ref{subfig:largeyaw} shows the trajectory for large initial yaw angle.  Table \ref{tab:costcom} compares the cost value for the two different cases.
\begin{table}[h]
    \centering
    \begin{tabular}{|c|c|c||c|c|}
    \hline 
  \multicolumn{3}{|c|}-  &{Small yaw } &{Large yaw }\\
   \hline
        $\mathbf{N}$ & Timesteps & Iterations & $J$ $\times 10^7$ &   $J$ $\times 10^7$ \\
        \hline 
       1  & 300 & 2 & 1.25 & 5.810 \\ \hline
       2  & 528 & 3  & 1.16 &  3.71 \\ \hline 
       3  & 820 & 3 & 1.12 &  3.14 \\ \hline 
    \end{tabular}
    \caption{Learning time and objective value for $\mathbf{N} = 1, 2, 3$}\vspace{-0.3cm}
    \label{tab:costcom}
\end{table}
From Figs. \ref{subfig:smallyaw},\ref{subfig:largeyaw} and Table \ref{tab:costcom} we see that not only does the objective value decrease with increasing $\mathbf{N}$, but for larger initial yaw angle, the nonlinearity dominates in the system, and the trajectories and cost for larger $\mathbf{N}$ show significant improvement as compared to the smaller yaw angle case.

\subsubsection{Structured RL control}
We next illustrate the performance of our proposed structured Carleman controller on the four tugboat model. The desired structure is shown in Fig. \ref{subfig:commstruc} where the bidirectional links between boat 1 and boat 4, boat 2 and boat 4 have been disabled.  Fig. \ref{subfig:xtstruc} shows the trajectories for the dense learnt controller (solid line) and the structured controller (dashed line) for $\mathbf{N}=2$. {\color{black} From Fig. \ref{subfig:xtstruc}, we can see that the closed-loop trajectories post the learning phase do not significantly differ from each other, even though the communication topology has been substantially reduced. However, it is important to note that if more than two links in addition to that shown in Fig. \ref{subfig:commstruc} are disabled, it is found that the proposed algorithm fails to converge to a stabilizing controller, supporting the argument that a controller for the desired topology must exist for the algorithm to work.}

\begin{figure}[h]
     \centering
     \begin{subfigure}[b]{0.49\linewidth}
         \centering
         \includegraphics[width=\linewidth,height=4cm]{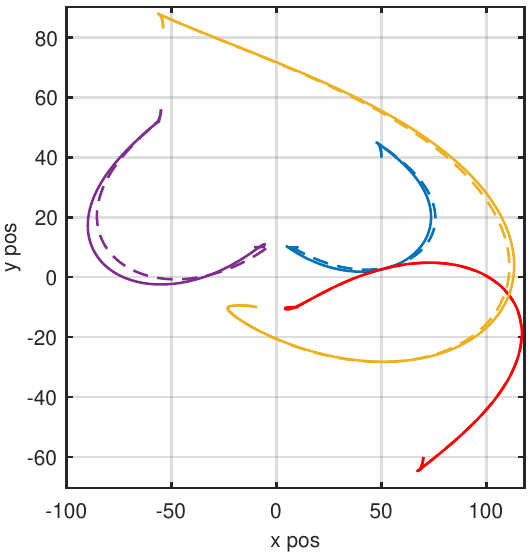}
         \caption{Closed-loop trajectory}
         \label{subfig:xtstruc}
     \end{subfigure}
     \hfill
     \begin{subfigure}[b]{0.49\linewidth}
         \centering
         \includegraphics[width=\textwidth,height=4cm]{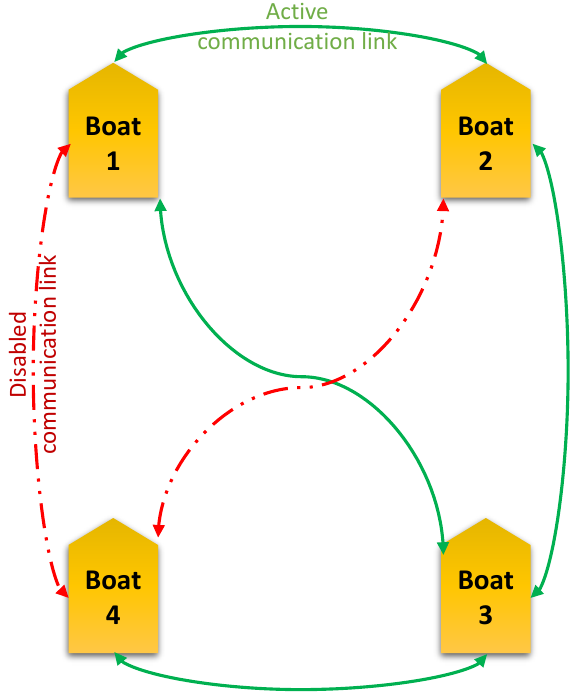}
         \caption{Communication structure}
         \label{subfig:commstruc}
     \end{subfigure}
        \caption{Structured Carleman RL for $\mathbf{N}=2$}
        \label{fig:structure}
\end{figure}

\subsubsection{Sparsity}
We next show the efficacy of the sparsity promoting algorithm \ref{alg:sparse} on our Carleman RL controller. Fig. \ref{fig:sparsity} shows the zoomed-in trajectories for the four tugboats at the final location for different values of $\gamma$ indicating increasing sparsity {\color{black} in the learnt closed-loop $N=2$ control gain. Note that this sparsity in terms of the communication links, but in terms of reducing the information being shared over the existing communication links.} It can be seen that for increasing $\gamma$, the trajectories start significantly deviating from the dense control. For this specific example, for $\gamma > 0.25$, the closed-loop system does not converge to the desired formation. 

\begin{figure}[h]
    \centering
    \includegraphics[width=\linewidth]{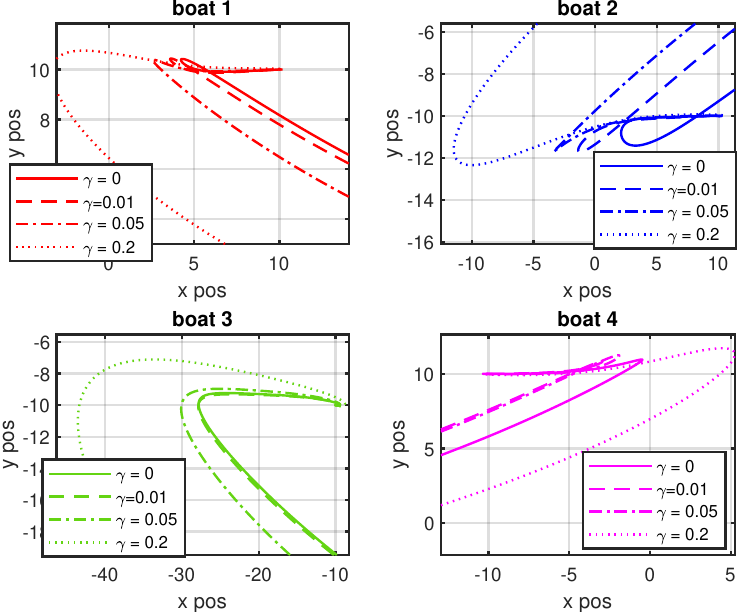}
    \caption{Sparse Carleman controller for varying $\gamma$}
    \label{fig:sparsity}
\end{figure}

Next, Fig. \ref{fig:badnwidth} shows the bandwidth requirements and objective function values for different $\gamma$. The max bandwidth for each bidirectional link is 12. The total system bandwidth is $12 \times 6 = 72$. Note that a state is transmitted if the corresponding entry of the control matrix is non-zero. Even if the entry in $K$ corresponding to the linear state is zero, but the Carleman states is non-zero, then the state needs to be transmitted. Fig. \ref{fig:badnwidth} shows that for an increasing value of $\gamma$, the total bandwidth requirement reduces. The trade-off, however, is in the performance, which can not only be seen in the difference in closed-loop trajectories in Fig. \ref{fig:sparsity}, but also in the increasing values of the objective function $J$ in Fig. \ref{fig:badnwidth}.

\begin{figure}[h]
    \centering
    \includegraphics[width=\linewidth,height=3.64cm]{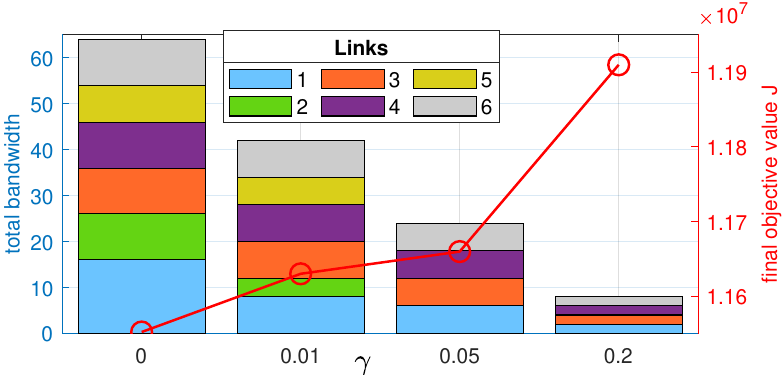}
    \caption{Total bandwidth requirement vs sparsity}
    \label{fig:badnwidth}
\end{figure}

\section{Conclusions and Future work}
We presented a data-driven approach for optimal control of unknown nonlinear systems. We use Carleman linearization to represent the nonlinear system in infinite-dimensional space, and develop policy iteration based methods for closed-loop control with any order of truncation. We also derive two variants of the Carleman RL design where the communication between various nonlinear subsystems in a network can be made structured and sparse, while still guaranteeing closed-loop stability. We compare our method with neural network based methods to show how our approach requires less learning time with no significant impact on closed-loop performance. For future work, we would like to extend this design to generic nonlinear system models that may not be input-affine.




\end{document}